\RequirePackage{fix-cm}
\documentclass[smallextended, final, a4paper]{svjour3}     
\smartqed  
\usepackage{amssymb,amsmath,graphicx}
\usepackage{times}
\usepackage[numbers]{natbib}
\usepackage{hyperref}

\newcommand{\objSet}{\mbox{{\cal O}}}
\newcommand{\nodeSet}{V}

\newcommand{\nodePref}[1]{\ge_{#1}}
\newcommand{\nodePrefEq}[1]{=_{#1}}
\newcommand{\nodePrefStrict}[1]{>_{#1}}
\newcommand{\placePref}[1]{\succeq_{#1}}
\newcommand{\placePrefStrict}[1]{\succ_{#1}}
\newcommand{\placePrefEq}[1]{=_{#1}}

\newcommand{\ppad}{\mbox{{\sc PPAD}}}
\newcommand{\pls}{\mbox{{\sc PLS}}}
\newcommand{\NP}{\mbox{{\sc NP}}}
\newcommand{\Poly}{\mbox{{\sc P}}}

\newcommand{\FSPPInstance}{{\cal I}}

\newcommand{\geqprefer}{\ge}

\newcommand{\pairPref}[1]{\sqsupseteq_{#1}}
\newcommand{\pairPrefStrict}[1]{\sqsupset_{#1}}
\newcommand{\pairPrefEq}[1]{=_{#1}}

\newcommand{\junk}[1]{}
\newcommand{\auxSet}{W}
\newcommand{\lca}{\mbox{lca}}

\newcommand{\ultraPairPref}[1]{\sqsupseteq_{#1}}
\newcommand{\ultraPairPrefStrict}[1]{\sqsupset_{#1}}

\newcommand{\prefGamePref}[1]{\ge_{#1}}
\newcommand{\prefGamePrefStrict}[1]{>_{#1}}

\newcommand{\costNoArgs}{d}
\newcommand{\cost}[2]{d_{#1#2}}
\newcommand{\sumUtility}[1]{U_s(#1)}

\newcommand{\internalSet}{I}
\newcommand{\ancestors}[1]{A_{#1}}

\newcommand{\cc}{\mbox{{\sc CSR}}}
\newcommand{\icc}{\mbox{{\sc CSR}}}
\newcommand{\fcc}{\mbox{{\sc F-CSR}}}
\newcommand{\fspp}{\mbox{{\sc FSPP}}}

\newcommand{\hold}{\mbox{hold}}
\newcommand{\holdPrime}{\mbox{hold'}}
\newcommand{\serve}{\mbox{serve}}
\newcommand{\servePrime}{\mbox{serve'}}

\newcommand{\fP}{\widetilde{P}}

\newcommand{\opt}{\mu}
\newcommand{\best}{b}

\newcommand{\BfPara}[1]{{\noindent {\bf #1.}}}

\newcommand{\twobin}{\mbox{{\sc 2BIN}}}
\newcommand{\twodirbin}{\mbox{{\sc 2DIR-BIN}}}
\newcommand{\exacttwodirbin}{\mbox{{\sc EXACT-2DIR-BIN}}}
\newcommand{\evencycle}{\mbox{{\sc EVEN-CYCLE}}}
\newcommand{\lbl}[1]{\ell(#1)}

\newcommand{\auxsrv}[2]{srv_#2(#1)}
\newcommand{\ItPara}[1]{{\noindent {\it #1.}}}
\newcommand{\DstList}{{\cal D}}

\newcommand{\beq}{\begin{equation}}
\newcommand{\eeq}{\end{equation}}
\newcommand{\beas}{\begin{eqnarray*}}
\newcommand{\eeas}{\end{eqnarray*}}

\newcommand{\node}[1]{#1}

\newcommand{\object}[1]{#1}

\newcommand{\distsrv}{d_{srv}}

\newcommand{\rate}[2]{r_{#1}(#2)}

\journalname{}
\begin{document}

\title{Cache Me If You Can: Capacitated Selfish Replication Games in
  Networks\thanks{Gopalakrishnan and Karuturi were partially supported by a generous gift from Northeastern University alumnus Madhav Anand.  This work was also partially supported by NSF grants CCF-0635119 and CNS-0915985.  A preliminary version of this work appeared in {\em LATIN} 2012~\cite{GKKPRS12}.}}

\author{Ragavendran Gopalakrishnan \and Dimitrios Kanoulas* \and Naga Naresh Karuturi \and C. Pandu Rangan \and Rajmohan Rajaraman \and Ravi Sundaram}

\institute{*Dimitrios Kanoulas (corresponding author) \at
           Istituto Italiano di Tecnologia,\\
           Via Morego, 30,\\
           Genoa, 16163, Italy \\
           Tel.: +1-617-971-8157\\
           \email{Dimitrios.Kanoulas@iit.it}}

\date{Received: date / Accepted: date}

\maketitle

\begin{abstract}
In Peer-to-Peer (P2P) network systems, content (object) delivery between nodes is often required.  One way to study such a distributed system is by defining games, which involve selfish nodes that make strategic choices on replicating content in their local limited memory (cache) or accessing content from other nodes for a cost.  These Selfish Replication games have been introduced in~\cite{ChunCWBPK04} for nodes that do not have any capacity limits, leaving the capacitated problem, i.e. Capacitated Selfish Replication (\cc) games, open.

In this work, we first form the model of the \cc\ games, for which we perform a Nash equilibria analysis.  In particular, we focus on hierarchical networks, given their extensive use to model communication costs of content delivery in P2P systems.  We present an exact polynomial-time algorithm for any hierarchical network, under two constraints on the utility functions: 1) ``Nearer is better'', i.e. the closest the content is to the node the less its access cost is, and 2) ``Independence of irrelevant alternatives'', i.e. aggregation of individual node preferences.  This generalization represents a vast class of utilities and more interestingly allows each of the nodes to have simultaneously completely different functional forms of utility functions.  In this general framework, we present \cc\ games results on arbitrary networks and outline the boundary between intractability and effective computability in terms of the network structure, object preferences, and the total number of objects.  Moreover, we prove that the problem of equilibria existence becomes NP-hard for general \cc\ games.  By adding some constraints in the number of objects and their preferences, we show that the equilibrium can be found in polynomial time.  Finally, we introduce the fractional version of \cc\ games (\fcc) to represent content distribution.  We prove that equilibrium exists for every \fcc\ game, but it is \ppad-complete.
\end{abstract}

\section{Introduction} \label{sec:intro}
Consider a P2P network for sharing movies (objects) among multiple users (nodes).  Due to limited disk space, the movies can be stored either locally or obtained from other users in some cost.  The storing decisions affect everyone that uses this service.  A natural question is to predict the movie collection stability in your friends network (i.e. equilibrium) and your satisfaction from them (i.e. access cost), when users act selfishly.  Similarly, in the new wireless 4G services, users will not only consume different apps, but will also provide apps to their network through personal communications and computing devices.  In such a network, the question is whether storing apps will lead to a situation of endless churn or could there be an equilibrium?

Content delivery and caching in P2P networks can be studied in a game-theoretic framework.  In this work, we study Capacitated Selfish Replication (\cc) games as an abstraction of the above network scenarios.  In \cc\ games the strategic agents, or players, are nodes in a network that act selfishly.  The nodes have some object preferences and bounded storage space, i.e. caches, to store a limited number of content copies.  Each node in cooperation with other nodes can serve access requests for the objects that are stored in its cache.  However, the set of objects which a node chooses to store in its cache is from one side solely based on its own utility function (notice that this does not prevent the players to use the same utility function for the whole network) and from the other side based on where objects of interest have been stored in the network.  Thus, each node in the model, has two potential actions for an object.  Either store a replica of the object in its limited cache, or access with some cost the object replica from a remote node.

Chun et al.~\cite{ChunCWBPK04} first introduced such a game-theoretic framework to analyse pure Nash equilibria in networks without cache capacities, but with some storage cost.  They left the capacitated version of the problem open.  The main interest of the \cc\ games in more recent works is on {\em hierarchical networks} that have been extensively used to model communication content delivery costs in P2P networks~\cite{Garces2013}.  Ultrametric models for content delivery networks~\cite{KLLLLP97} and cooperative caching in hierarchical networks~\cite{LWY93, TDVK99, KPR01, KD02} are just some examples.  The best results on \cc\ games for hierarchical networks~\cite{LaoutarisSelfish,PollatosSelfish} are about the existence of a Nash equilibrium for a generalized one-level hierarchical network, using the sum utility function for which each node is based on a weighted sum of the cost of accessing the objects.

\subsection{Our results}
In this paper, we first introduce the basic model of the Capacitated Selfish Replication (\cc) games in Section~\ref{sec:sum_model}.  This includes the definition of the nodes (players) and objects, the formulation of the cost functions for a node accessing objects in the network, the object replication strategies among nodes, as well as the basic formulation of the network.  The main focus is on the study of Nash equilibria existence and computability for a set of \cc\ games variants.  In particular, we introduce a polynomial-time Nash equilibrium method for hierarchical networks, given their extensive use to model communication costs of content delivery in P2P systems.  We address the following three problems, including their computational complexity:
\begin{itemize}
  \item Does pure Nash equilibrium exist in a \cc\ game, for hierarchical networks?
  \item Does pure Nash equilibrium exist in a \cc\ game, for general undirected networks, setting specific restrictions on the number of objects and the utility/cost functions?
  \item Does pure Nash equilibrium exist, when the objects can be split in fractions, i.e. \fcc games?
\end{itemize}
Note that in all the games, we assume that all the pieces of content, i.e. objects, have the same size, as considered in prior works~\cite{ChunCWBPK04, LaoutarisSelfish, PollatosSelfish, ABPZ13}.  Otherwise, the problem becomes NP-hard even for computing the best response of a player (node) as a generalization of the well-known knapsack problem.

In Section~\ref{sec:sum_ultrametric} we present our main algorithm, which extends and resolves the open problem that was defined in~\cite{LaoutarisSBS06, LaoutarisSOSB07, PollatosSelfish}.  In particular, it has been proved~\cite{LaoutarisSelfish,PollatosSelfish} that \cc\ games for hierarchical networks have a Nash equilibrium in the case of a generalized $1$-level hierarchy, when the utility function is a function of the costs sum of accessing replicated objects in the network.  We introduce an exact polynomial-time algorithm for Nash Equilibrium computation in any hierarchical network.  We use a novel technique which we name ``fictional players\footnote{not to be confused with ``fictitious play''~\cite{FudenbergLevine98} which involves learning}'' method.  We show that using this method we can extend to a general framework of natural preference orders that are entirely arbitrary, but follow two natural constraints: ``Nearer is better'', i.e. the closest the content is to the node the less its access cost is and ``Independence of irrelevant alternatives'', i.e. the aggregation of individual node preferences.  This generalization represents a vast class of utility functions and more interestingly allows each of the nodes to have simultaneously completely different functional forms of utility functions.  The method introduces and iteratively eliminates fictional players in a controlled fashion, maintaining a Nash equilibrium at each step.  In the end, the desired equilibrium for the entire network is realized without any fictional players left in the network.  Even though the analysis is specified in the context of the sum utility function to elucidate the technique of fictional players, we then abstract the central requirements for our proof technique.  In particular, we develop a general framework of \cc games with ordinal preferences, for which a larger class of utility functions could be used as extension to the above result.

In Section~\ref{sec:model}, we present the general \cc\ games framework in terms of the utility preference relations and node preference orders.  In particular, we consider the utility that is not just each node's specific numeric assignment for each objects placement, but a preference order each node has on object placements that satisfies two natural constraints: Monotonicity (or ``Nearer is better'') and Consistency (or ``Independence of irrelevant alternatives'').  In this way the method is extended to a vast class of utility functions, while nodes may simultaneously have utility functions of completely different functional forms.

After extending our hierarchical networks results to the broader class of utilities, in Sections~\ref{sec:exist}
and~\ref{sec:non-exist} we study general \cc\ games that have various network structures (directed or undirected), forms of object preferences (binary or general).  Intractability and effective computability of equilibria is delineated in terms of the network structure, object preferences, and the total number of objects.  The results are summarized in Table~\ref{tab:complexity}.  Most notable are the following results:
\begin{itemize}
  \item equilibria existence for general undirected networks with two objects, using the technique of fictional players
  \item equilibria existence for general undirected networks when object preferences are binary
  \item the problem of equilibria existence becomes NP-hard for general \cc\ games
  \item equivalence of finding equilibria in polynomial time for \cc\ games in strongly connected networks with two objects and binary object preferences, via a reduction to the well-studied even-cycle problem~\cite{RST99}.
\end{itemize}

\begin{table}
\begin{center}
\begin{tabular}{||l|l|l||}\hline
{\bf Object preferences and count} & {\bf Undirected networks} &  {\bf Directed networks} \\ \hline\hline
Binary, two objects & Yes, in \Poly~(\ref{sec:2-obj-metric}) & No, in \Poly~
(\ref{sec:dir2obj})\\ \hline
Binary, three or more objects & Yes, in \pls~(\ref{sec:01-metric}) & No, \NP-complete (\ref{sec:npc})\\ \hline
General, two objects & Yes, in \Poly~(\ref{sec:2-obj-metric}) & No, \NP-complete (\ref{sec:npc})\\ \hline
General, three or more objects & No, \NP-complete (\ref{sec:npc}) & No, \NP-complete (\ref{sec:npc})\\
& Hierarchical: Yes, in \Poly~(\ref{sec:ultrametric}) & \\ \hline\hline
\end{tabular}
\caption{{\small Equilibria existence and computability in \icc\ games.  Each cell (other than in the first row/column) first indicates whether equilibria always exist for a particular \icc\ games sub-class.  If so, the cell next indicates the complexity of determining an equilibrium; otherwise, it indicates the complexity of determining whether equilibria exist.  The relevant subsection appears in parentheses.}\label{tab:complexity}}
\end{center}
\end{table}

Finally in Section~\ref{sec:frac}, we introduce the fractional version of \cc\ games (\fcc) to represent content distribution using erasure codes.  In this framework, each node is allowed to store fractions of objects and can satisfy an object access request by retrieving any set of object fractions as long as these fractions sum to at least one.  We present a natural implementation of this framework via erasure codes (e.g. using the Digital Fountain approach~\cite{Byers98,Shok06}).  We prove that \fcc\ games always have equilibria and finding it is in \ppad.  However, we also show finding equilibria is \ppad-hard even for a sum-of-distances utility function.

\subsection{Related work}
Peer-to-Peer (P2P) networks have been used to model systems for sharing content and resources among the individual peers (such as the file systems~\cite{Kubiatowicz2000, Dabek2001, Rowstron2001, Saito2002}, web caches~\cite{Danzig1998, Fan1998}, or P2P caches~\cite{Iyer2002}).  Even though P2P networks have been extensively studied from a theoretical point of view, there are several open problems when rational peers have diverse and selfish interests~\cite{Feldman2005}.  

One of the most interesting problems is {\em caching}, i.e. holding copies of content in clients and servers.  Several research studies have considered data storing~\cite{Gribble2001, Yan2002}, self-stabilization~\cite{Ko2005}, dynamic replication~\cite{Rabinovich1999, Douceur2001, Tang2002}, and exchanging of content copies in a centralized manner~\cite{Li1999, Jamin2000, Qiu2001, Jamin2001}.   Research on capacitated caching has been also considerable as an optimization problem and various centralized and distributed algorithms have been presented for different networks in~\cite{LWY93, WJH97, KPR01, BRS08, ABPZ13}.  For instance, centralized optimization for the facility location problem has been studied in~\cite{Rosenwein1994}, including several approximations~\cite{Jain1999, Mettu2000, Mahdian2002}.  These frameworks usually ignore the fact that peers may make free choices that minimizes their content access cost, by not following usual instrumentation.

The caching problem that we study is in the intersection of game theory and computer science, that has been extensively studied the last decade~\cite{NisanRTV07, Tsaknakis2008}.  In~\cite{Papadimitriou94} Papadimitriou laid the groundwork for algorithmic game theory by introducing syntactically defined sub-classes of FNP with complete problems, \ppad\ being a notable such subclass.  Non-cooperative facility location games have attracted some small attention over the last decades.  For instance, in~\cite{Vetta2002}, the problem of Nash equilibrium for games that allowed players build nodes in remote locations, whereas in our case nodes hold fixed spaces for storing objects/content.  In~\cite{Goemans2006}, content distribution was studied, providing bounds on the approximated Nash equilibrium with respect to the price of anarchy and the convergence speed.  The difference in the game design lies in the fact that each node had cost limits for storing objects without considering network latencies.  The uncapacitated case of selfish caching games was introduced in~\cite{ChunCWBPK04}, in which nodes could store more pieces of content by paying for the additional storage.

We focus on the capacitated version which was left open by~\cite{ChunCWBPK04}, believing that limits on cache-capacity model an important real-world restriction.  Special cases of the integral \cc\ games version have been studied.  In~\cite{LaoutarisSelfish}, Nash equilibria were shown to exist in cases that nodes are equidistant from one another and a special centralized server holds all objects.  In~\cite{PollatosSelfish} the model is slightly extended to the case where special servers for different objects are at different distances.  Our results generalize and completely subsume all these prior cases of \cc\ games.  Market sharing games~\cite{GoemansLMT06} also consider caches with capacity, but differ to cc\ games since they are special cases of congestion games.  In this work we focus primarily on equilibria and our general framework of \cc games with ordinal preferences aligns more with the theory of social choice~\cite{arrow51social}; in this sense, we deviate from prior work~\cite{fabrikant03network, devanur05price} that is focused on the price of anarchy~\cite{koutsoupias99worst}.

Our work on \cc\ games in~\cite{GKKPRS12} has initiated various research lines and has been extended recently in different directions.  For instance, in~\cite{HG2013,HG2014} the selfish replication problem is studied for the case that nodes seek object placements with cache cooperation, and includes an experimental analysis.  Etesami et al. have extended our model in a series of papers~\cite{Etesami2017-book}.  In~\cite{Etesami2014PureNE} the Nash equilibrium algorithm for two resources is shown to converge faster and it is extended to arbitrary cache sizes for a polynomial time computation.  This is extended in~\cite{EB2015, EB2016-arxiv, EB2016-arxiv2}, where a quasi-polynomial algorithm is introduced to drive allocations whose total cost is within a constant factor of that in any pure-strategy Nash equilibrium, in games formed by undirected networks.  The price of anarchy for \cc\ games with binary preferences over general undirected networks has been studied in~\cite{EB2016, EB2017}, showing an upper bound of $3$.  In~\cite{Pacifici2016}, the caching problem is studied for operator-specific, non-linear, cost functions in games that form arbitrary peering graph topologies, while in~\cite{AEP2016} \cc\ games are studied for general undirected networks for which a randomized algorithm is introduced using a random tree search method to search for pure-strategy Nash equilibrium.

In related work, through a major breakthrough~\cite{DaskalakisGoldbergPapadimitriou06, ChenDengTengJACM}  it has been proven that 2-player Nash Equilibrium is \ppad-complete.  The PPAD-complete term is coming to occupy a role in algorithmic game theory analogous to NP-completeness in combinatorial optimization~\cite{GareyJohnson}, and thus we study the fractional version of the problem, where nodes can store parts of objects, while accessing the remaining part from other nodes.  In this setup we prove PPAD-completeness.

\section{A basic model for \cc\ games} \label{sec:sum_model}
We consider a set $\nodeSet$ of nodes (labeled $1$ through $n = |V|$) to form a network in which they share a collection $\objSet$ of unit-size objects.  We let $\cost{i}{j}$ denote $i$'s cost for accessing an object at $j$, for $i,j \in \nodeSet$; we refer to $\costNoArgs$ as the access cost function.  $j$ is node's $i$ {\em nearest}\/ node in a set $S$ of nodes, if $j \in S$ and $\cost{i}{j} \leq \cost{i}{k}$ for all $k \in S$.  Moreover, a given network is {\em undirected}\/ if $\costNoArgs$ is symmetric, i.e. if $\cost{i}{j} = \cost{j}{i}$ for all $i,j \in \nodeSet$.  An undirected network is {\em hierarchical}\/ if the access cost function forms an ultrametric, i.e. if $\cost{i}{k} \le \max \{\cost{i}{j}, \cost{j}{k}\}$ for all $i,j,k \in \nodeSet$.

The cache of each node $i$ is able to store a certain number of objects.  Node's $i$ placement is simply the set of objects stored at
$i$.  The strategy set of a given node is the set of all feasible
placements at the node.  A {\em global placement}\/ is any tuple $(P_i: i \in V)$, where $P_i \subseteq \objSet$ represents a feasible
placement at node $i$.  We are going to use $P_{-i}$ to denote the
collection $(P_j: j \in V \setminus \{i\})$, for convenience.  We will also often use $P = (P_i, P_{-i})$ to refer to a global placement.  Moreover, we also assume that $V$ includes a \emph{ server} node that has the capacity to store all the objects.  In this way it is ensured that at least one copy of every object is present in the system; this assumption is without loss of generality given that the access cost of every node to the server can be set an arbitrarily large number.

\smallskip \BfPara{\icc\ Games} Each node in our game-theoretic model, attaches a utility to each global placement.  We assume that each node $i$ has a weight $r_i(\alpha)$ for each object $\alpha$ representing the rate at which $i$ accesses $\alpha$.  We define the {\em sum utility function}\/ $\sumUtility{i}$ as follows: $\sumUtility{i}(P) = - \sum_{\alpha \in \objSet} r_i(\alpha) \cdot
\cost{i}{\sigma_i(P,\alpha)}$, where $\sigma_i(P,\alpha)$ is $i$'s
nearest node holding $\alpha$ in $P$.  A \icc\ game is a tuple $(\nodeSet, \objSet, \costNoArgs, \{r_i\})$.  This work focuses on {\em pure Nash equilibria}\/ (henceforth, simply {\em equilibria}) of the \cc\ games.  Such a \icc\ game equilibrium instance is a global placement $P$ such that for each $i \in \nodeSet$ there is no placement $Q_i$ such that $\sumUtility{i}(P) < \sumUtility{i}(Q)$.

\smallskip
\BfPara{Unit cache capacity} In this work, we assume that all objects are of identical size.  Under this assumption, we now argue that it is sufficient to consider the case where each node's cache holds exactly one object.  Consider a set $\nodeSet$ of nodes in which the cache of node $i$ can store $c_i$ objects.  Let $\nodeSet'$ denote a new set of nodes which contains, for each node $i$ in $\nodeSet$, new nodes $i_1, i_2, \ldots, i_{c_i}$, i.e., one new node for each unit of the cache capacity of $i$.  We extend the access cost function as follows: $\cost{j_\ell}{i_k} = \cost{j}{i}$ for all $1 \le \ell \le c_j$, $1 \le k \le c_i$ and $\cost{i_\ell}{i_k} = 0$ for all $1 \le \ell < k \le c_i$, for each node $i \in \nodeSet$.

We consider an obvious onto mapping $f$ from placements in $\nodeSet'$ to those in $\nodeSet$.  Given placement $P'$ for $\nodeSet'$, we set $f(P') = P$ where $P_i = \cup_{1 \le k \le c_i} P'_{i_k}$.  This mapping ensures that $\sumUtility{i}(P') = \sumUtility{i}(P)$, giving us the desired reduction. Thus, in the remainder of the paper, we assume that every node in the network stores at most one object in its cache.

\section{Hierarchical networks} \label{sec:sum_ultrametric}
In this section, we present a polynomial-time equilibria construction for \icc\ games on hierarchical networks.  We can represent any hierarchical network by a tree $T$ in such a way that the node set $\nodeSet$ is the set of its leaves.  Every internal node $v$ has a label $\lbl{v}$ such that:
\begin{enumerate}
  \item if $v$ is an ancestor\footnote{We let each node be both descendant and ancestor of itself.} of $w$ in $T$, then $\lbl{v} \geq \lbl{w}$
  \item for any $i, j \in \nodeSet$, $\cost{i}{j}$ is given by $\lbl{\lca(i,j)}$, where $\lca(i,j)$ is the least common ancestor of nodes $i$ and $j$~\cite{KLLLLP97,KPR01}.
\end{enumerate}

Fig.~\ref{Fig:sec3} illustrates a simple example for a hierarchical network tree with two internal nodes and three leaf nodes, with the corresponding label relationships, the least common ancestors, and the access costs.

\begin{figure}[ht]
  \centering
  \includegraphics[width=\textwidth]{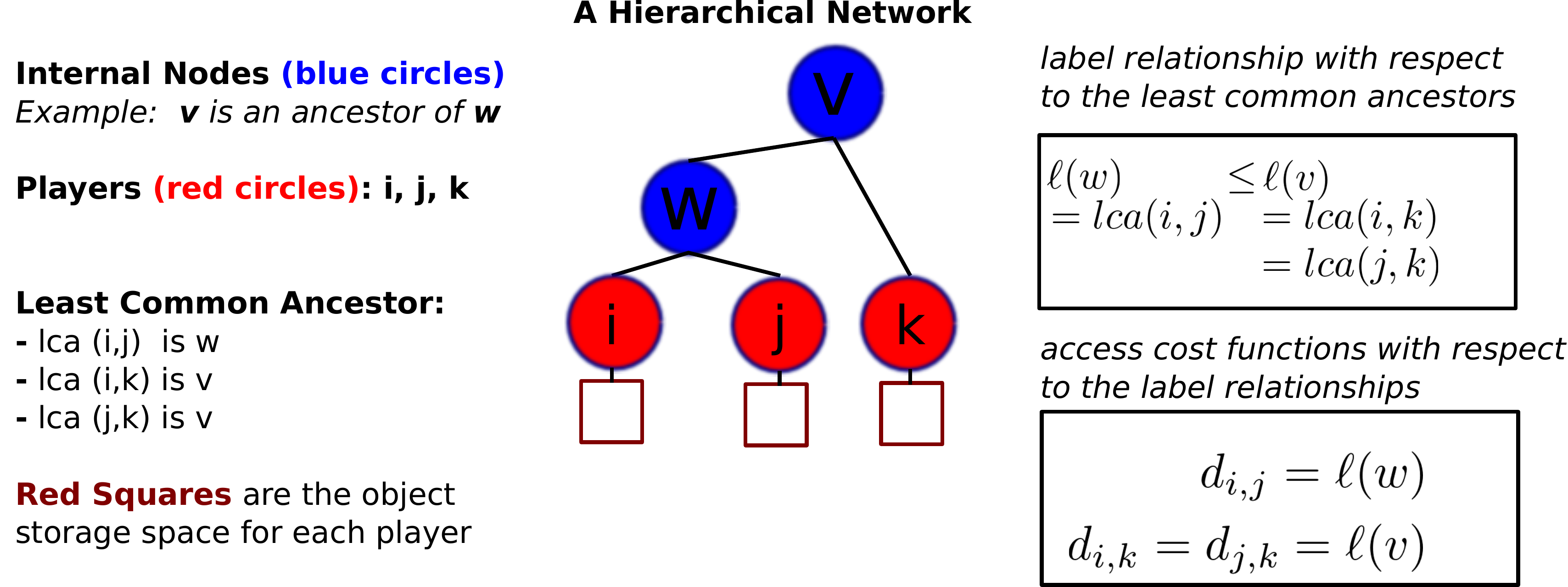}
\caption{\small A simple example of a hierarchical network tree with two internal nodes ($\lbl{v}$ and $\lbl{w}$) and three leaf nodes (i, j, and k).  The label relationships, the least common ancestors, and the access costs are described.}
\label{Fig:sec3}
\end{figure}

\smallskip \BfPara{Fictional players} The proposed algorithm requires the introduction of the {\em fictional player} notion.  A {\em fictional $\alpha$-player}\/ for an object $\alpha$ will be a new node which stores $\alpha$ in any equilibrium.  In particular, for any fictional $\alpha$-player $\ell$, $r_\ell(\alpha)$ is $1$ and $r_\ell(\beta)$ is $0$ for any $\beta \neq \alpha$.  In a particular hierarchy each fictional player is introduced as a leaf; our method determines the exact locations in the hierarchy.  The access cost function for each fictional player is naturally extended using the hierarchy and the labels of the internal nodes.  We let ``node'' denote both the elements of $V$ and fictional players.

\smallskip \BfPara{A preference relation} The object weights for each node $i$ in a hierarchical network induce a natural preorder $\ultraPairPref{i}$ among elements of $\objSet \times \ancestors{i}$, where $\ancestors{i}$ is the set of proper ancestors of $i$ in $T$.  In particular, we define $(\alpha, v) \ultraPairPrefStrict{i} (\beta, w)$ whenever $r_i(\alpha)\cdot \lbl{v} > r_i(\beta) \cdot \lbl{w}$.  In words, in hierarchical networks there is a total preorder in the objects-nodes preferences, which is used during the algorithm to define a potential function, when nodes are playing their best responses.  For instance, $(\alpha, v) \pairPrefStrict{i} (\beta, w)$ means that if $i$ needs to place either $\alpha$ or $\beta$ in its cache, and the least common ancestor of $i$ and the most $i$-preferred node in $\nodeSet \setminus \{i\}$ holding $\alpha$ (resp., $\beta$) is $v$ (resp., $w$), then $i$ prefers to store $\alpha$ over $\beta$. Fig.~\ref{Fig:sec3_2} illustrates an example of node $i$ that will prefer to store object $\alpha$ that is stored further than object $\beta$ and with a higher cost, due to the total preorder.

\begin{figure}[ht]
  \centering
  \includegraphics[width=0.5\textwidth]{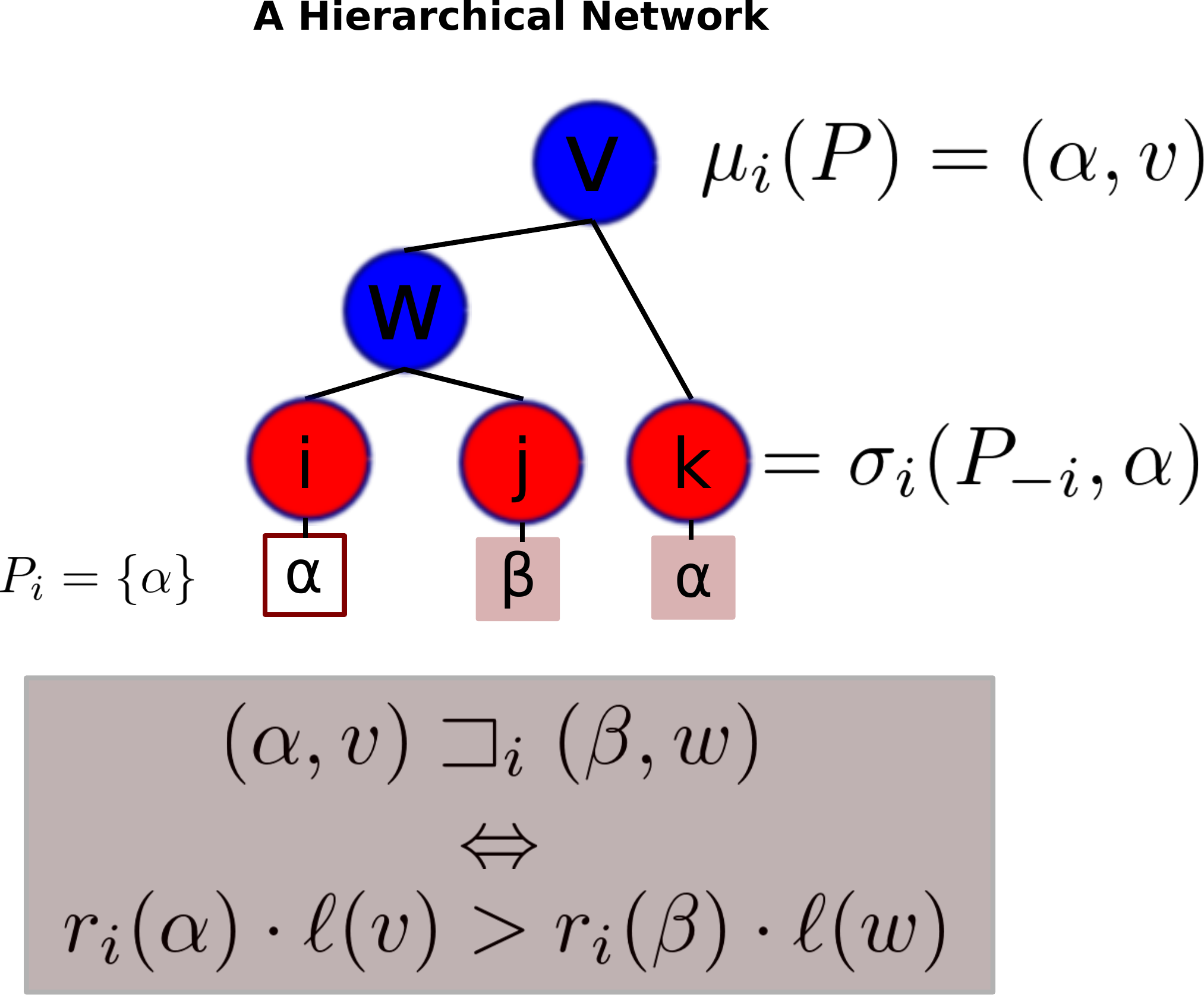}
\caption{\small A simple example of a hierarchical network tree with two internal nodes ($\lbl{v}$ and $\lbl{w}$) and three leaf nodes (i, j, and k).  The preference relation of node $i$ is presented.}
\label{Fig:sec3_2}
\end{figure}

To express any player's best response in terms of these preference relations, we define $\mu_i(P) = (\alpha, v)$, where $P_i = \{\alpha\}$ and $v$ is $\lca(i,\sigma_i(P_{-i},\alpha))$, where $\sigma_i(P_{-i}, \alpha)$ denotes $i$'s nearest node in the set of nodes holding $\alpha$ in $P_{-i}$.  For instance, in Fig.~\ref{Fig:sec3_2} $\sigma_i(P_{-i}, \alpha)$ is node $k$ (the nearest node holding $\alpha$), while $P_i = \{\alpha\}$ (node $i$ is holding $\alpha$) and the  thus, these two information can be denoted as $\mu_i(P) = (\alpha, v)$, where $v$ is the least common ancestor between nodes $i$ and $k$.

Given, the aforementioned definitions, we can now express the best response of a player in terms of the preference relations in the following Lemma.  This is needed in Lemma~\ref{lem:ultra} to prove the existence of an equilibrium at each step of the algorithm.

\begin{lemma} \label{lem:sum_pair_pref}
A best response $P_i$ of a node $i$ for a placement $P_{-i}$ of $V \setminus \{i\}$ is $\{\alpha\}$ where $\alpha$ maximizes $(\gamma, \lca(i,\sigma_i(P_{-i}, \gamma)))$, over all objects $\gamma$, according to $\ultraPairPref{i}$.
\end{lemma}

\begin{proof}
For a given placement $P$ with $P_i = \{\alpha\}$, $\sumUtility{i}(P)$
equals
\[
- \sum_{\gamma \neq \alpha} r_i(\gamma)
\ell(\lca(i,\sigma_i(P_{-i}, \gamma))),
\] which can be rewritten as
\[
- (\sum_{\gamma \in \objSet} r_i(\gamma) \ell(\lca(i,\sigma_i(P_{-i},
\gamma)))) + r_i(\alpha) \cdot \ell(\lca(i,\sigma_i(P_{-i},
\alpha))).
\]  Thus, $\{\alpha\}$ is a best response to $P_{-i}$ if and only if $\alpha$ maximizes $r_i(\gamma)$ $\cdot$ $\ell(\lca(i,\sigma_i($ $P_{-i}, \gamma))$ over all objects $\gamma$, while the desired claim follows from the definition of $\ultraPairPref{i}$.
\end{proof}

\smallskip \BfPara{The algorithm} In the beginning of the algorithm we introduce a set of fictional players, maintaining in the same time the invariant that the current global placement in this hierarchy is an equilibrium.  We then proceed by removing existing or adding new fictional players, tweaking in a particular way their set and locations, in such a way that at each step we guarantee an equilibrium.  The algorithm terminates when all the fictional players are removed in the desired equilibrium state.  Let $\auxSet_t$ and $P^t$ denote the set of fictional players and equilibrium, respectively, at the start of step $t$ of the algorithm.  

\smallskip
{\sl Initialization.} We create an initial set $W_0$ by adding a fictional $\alpha$-player as a leaf child of $v$, for each object $\alpha$ and internal node $v \in T$.  In the initial equilibrium $P^0$ for each fictional $\alpha$-player $i$ we have $P^0_i = \{\alpha\}$, i.e. each node $i \in V$ plays its best response.  By definition, it is clear that each fictional player is in equilibrium.  Moreover, for every $\alpha$, every $i \in V$ has a sibling
fictional $\alpha$-player.  Thus, the best response of every $i \in V$ does not depend on the placement of nodes in $V \setminus \{i\}$, which implies that $P^0$ is an equilibrium.

\smallskip {\sl Algorithm's $t$ step.} For the node set $V \cup W_t$ (the original nodes and the fictional ones) we fix an equilibrium $P^t$.  If $W_t$ is empty, i.e., no fictional player remained, we are done.  Otherwise, we select a fictional node $j$ in $W_t$.  Let $P^t_j = \{\alpha\}$ and $\mu_j(P^t) = (\alpha, v)$, i.e. the fictional player $j$ holds object $\{\alpha\}$ and the closest node that holds object $\{\alpha\}$ is through the internal node $v$.  We let $S$ be the set of all nodes $i \in V$ such that $(\alpha, v) \ultraPairPrefStrict{i} \mu_i(P^t)$, i.e. the closest node that holds object $\{\alpha\}$ (except itself) is through the internal node $v$.  We consider two cases for computing a new set of fictional players $W_{t+1}$ and a new global placement $P^{t+1}$ such that $P^{t+1}$ is an equilibrium for $V \cup W_{t+1}$:

\textit{$S$ is empty} (there is a node holding the object closer than through the internal node $v$ and thus the fictional node $j$ is not affecting the strategy).  We remove the $j$ fictional player from $W_t$ and the hierarchy.  For the remaining nodes the placement remains as is.  In this way $W_{t+1} = W_t - \{j\}$ (the fictional player is removed) and $P^{t+1}$ is the same as $P^t$ (since the fictional player j was not affecting any other node's best response strategy), but $P^{t+1}_j$ is no longer defined, since $j$ is removed.

\textit{$S$ is nonempty} (some nodes are accessing object $\alpha$ from the fictional player $i$).  We select a node $i \in S$ such that $\lca(i,j)$ is lowest among all nodes in $S$ (in this way no other node is affected from the change in the strategy of i) and we let $P^t_i = \{\beta\}$.  We set $P^{t+1}_i = \{\alpha\}$, remove the fictional $\alpha$-player $j$ from $W_t$, and add a new fictional $\beta$-player $\ell$ as a leaf sibling of $i \in T$ (in this way the player will be in equilibrium by accessing $\beta$ from the new fictional player).  In this way $P^{t+1}_\ell = \{\beta\}$, while for every other node $j$ we set $P^{t+1}_j = P^t_j$.  Finally, we set $W_{t+1} = (W_t \cup \{\ell\}) \setminus \{j\}$, removing from the node set the removed fictional player and adding the new one.

An example of the steps is illustrated in Fig.~\ref{fig:Fig_3}.  Next, in Lemma~\ref{lem:ultra} we prove why at every step $t$, as described above, we have an equilibrium.

\begin{figure}[ht]
\centering
\includegraphics[scale=0.1]{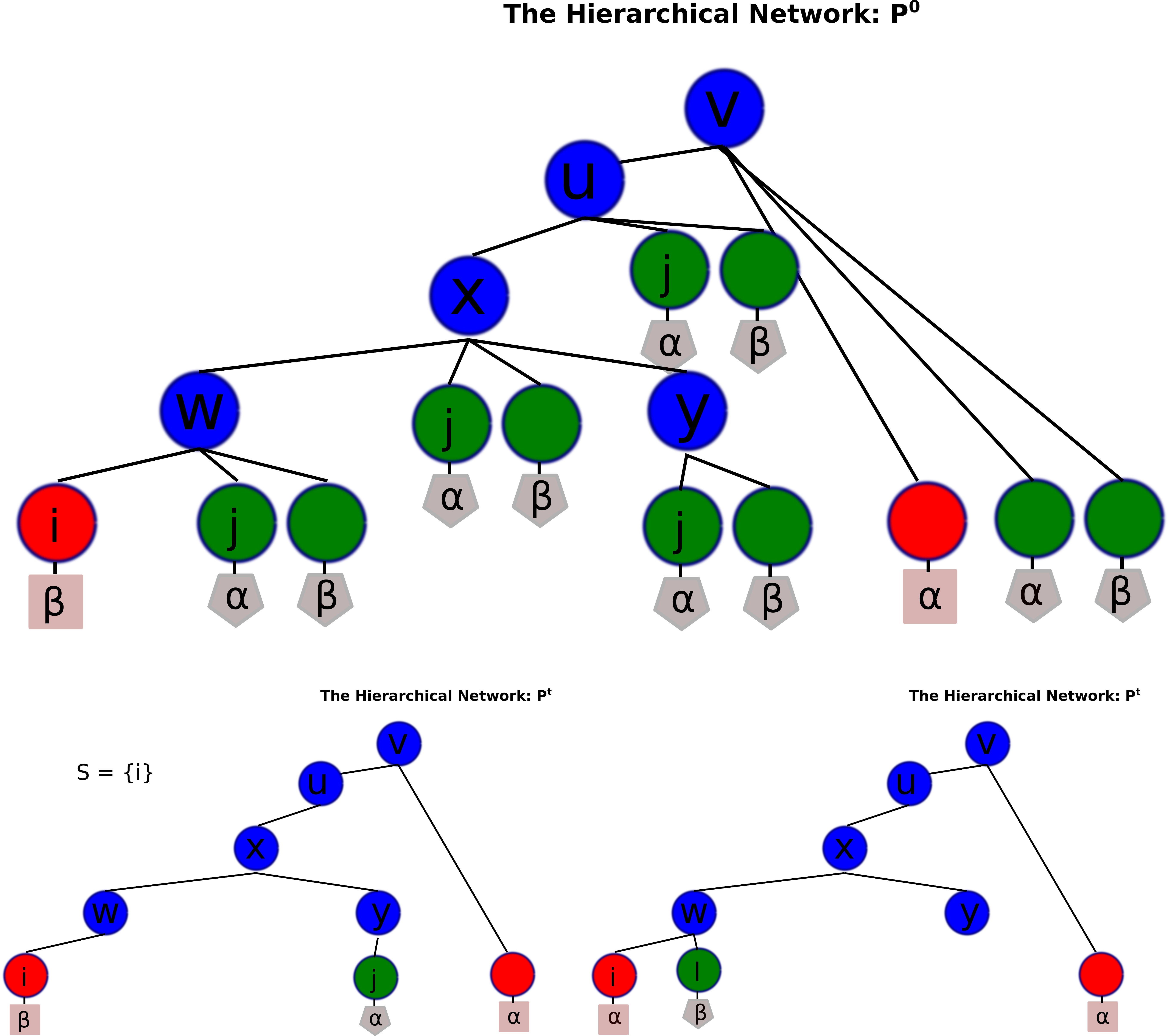}
\caption{\small Illustrating the algorithm for a simple hierarchical network.}
\label{fig:Fig_3}
\end{figure}

\begin{figure}[ht]
\centering
\includegraphics[scale=0.15]{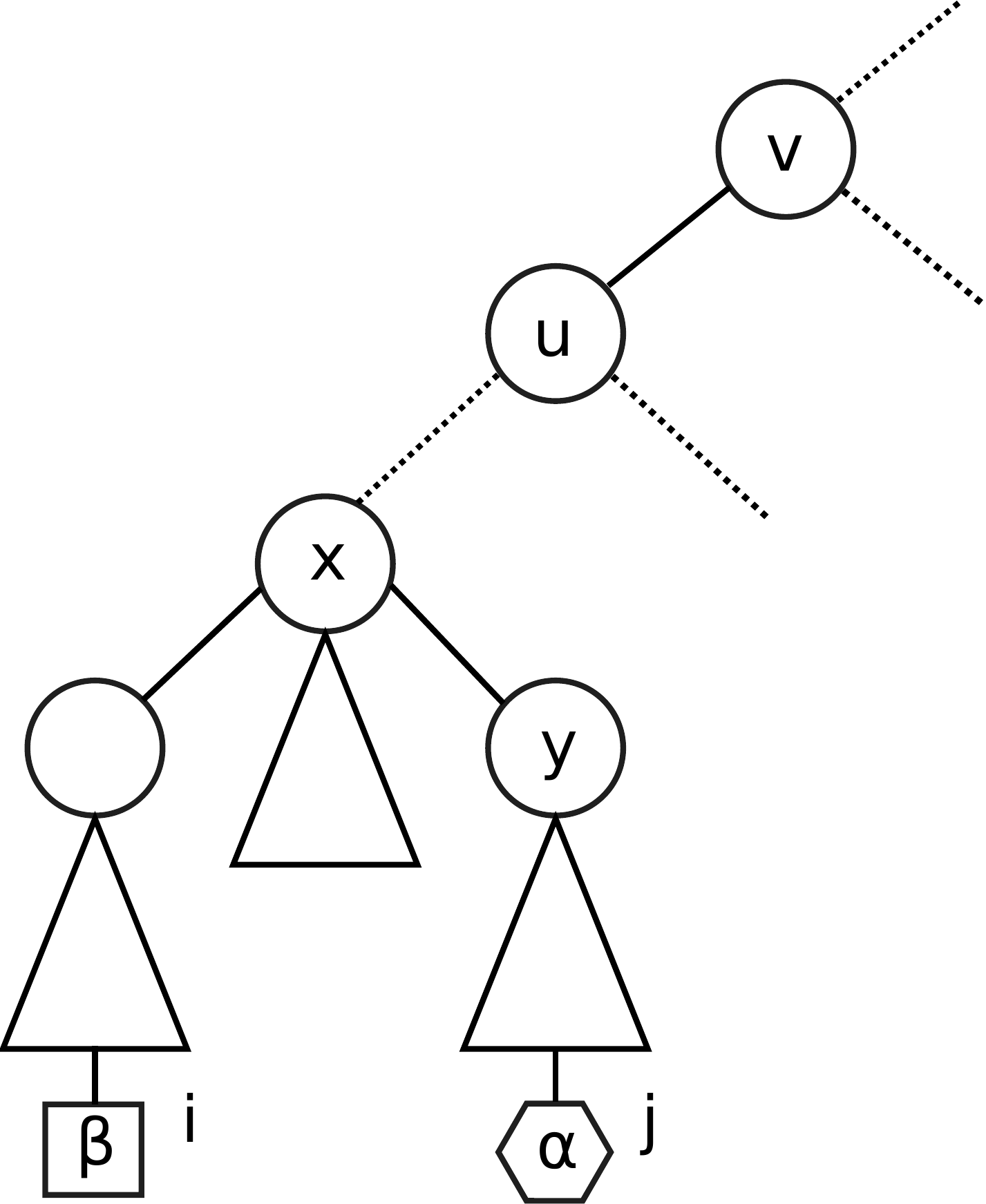}
\caption{\small Illustrating the analysis for hierarchical networks; referred to in the proof of Lemma~\ref{lem:ultra}.  The square is a node $i$ in $V$ holding object $\beta$, and the hexagon is a fictional $\alpha$-player. \label{fig:examples}}
\end{figure}

\junk{
\begin{figure}[ht]
\begin{minipage}[b]{0.3\textwidth}
	\centering
	\resizebox{0.5\textwidth}{!}
	{\includegraphics{npc_g_n.pdf}}
\end{minipage}
\begin{minipage}[b]{0.3\textwidth}
	\centering
	\resizebox{0.6\textwidth}{!}
	{\includegraphics{npc_g_n1.pdf}}
\end{minipage}
	\begin{minipage}[b]{0.35\textwidth}
	\centering
	\resizebox{0.5\textwidth}{!}
	{\includegraphics{ultram.pdf}}
\end{minipage}
	\caption{\small (a) Left: an instance for directed graphs. (b)
          Center: an instance for undirected graphs; (c) Right:
          Illustrating the analysis for hierarchical networks; referred to in the proof of Lemma~\ref{lem:ultra}.  The
          square is a node $i$ in $V$ holding object $\beta$, and the
          hexagon is a fictional
          $\alpha$-player. \label{fig:examples}}
\end{figure}}

\begin{lemma}
\label{lem:ultra}
Consider step $t$ of the algorithm.  If $P^t$ is an equilibrium for
$\nodeSet \cup W_t$, then the following statements hold.

\begin{enumerate}
  \item For every node $k$ in $V \cup W_{t+1}$, $P^{t+1}_k$ is a best response to $P^{t+1}_{-k}$.
  \item For every node $k$ in $V \cup W_{t+1}$, $\mu_k(P^{t+1}) \ultraPairPref{k} \mu_k(P^t)$.
  \item We have $|W_{t+1}| \le |W_t|$.  Furthermore, either $|W_{t+1}| < |W_t|$ or there exists a node $i$ in $V$ such that $\ \mu_i(P^{t+1}) \ultraPairPrefStrict{i} \mu_i(P^t)$.
\end{enumerate}
\end{lemma}

\begin{proof}
Let $\alpha$, $v$, $S$, $i$, and $j$ be defined as described above in step $t$.  We first prove the first two lemma's statements.  We let $k$ be any node in $V \cup W_{t+1}$.  First, we consider the case that $\lca(k,j)$ is an ancestor of $v$.  In this case $k$ is not in the subtree rooted at the child $u$ of $v$ that contains $j$.  For any object $\gamma$, $\sigma_k(P^{t+1}_{-k}, \gamma) = \sigma_k(P^t_{-k}, \gamma)$ and $P_k^{t+1} = P^t_k$.  Statement 2 for $k$ is implied thus from the fact that $\mu_k(P^{t+1}) = \mu_k(P^t)$.  Since $P^t$ is in equilibrium, statement 1 also holds for $k$.
We then consider the case that $\lca(k,j)$ is a proper descendant of $v$.  in this case $k$ is in the subtree rooted at the child $u$ of $v$ that contains $j$.  There are two cases.

In the case that $S$ is empty, the fictional $\alpha$-player $j$ is removed.  In this way $j$ is not in $W_{t+1}$.  Moreover, there is no copy of $\alpha$ in the subtree rooted at $u$.  Given that no other object except $\alpha$ is created or removed, $\sigma_k(P^{t+1}_{-k}, \gamma) = \sigma_k(P^t_{-k}, \gamma)$ for $\gamma \neq \alpha$.  The second statement is established for $k$ by the fact that $\lca(k, \sigma_k(P^{t+1}_{-k}, \alpha)) = v$ and $\mu_k(P^{t+1}) =
\mu_k(P_t)$.  Since $S$ is empty, $\mu_k(P^t) \ultraPairPref{k} (\alpha,v)$.  The first statement for $k$ follows from Lemma~\ref{lem:sum_pair_pref} and the fact that $P^t_k$ is in equilibrium such that $P^{t+1}_k$ is a best response against $P^{t+1}_{-k}$.

In the second case that $S$ is not empty, we let $i$ be a node in $S$ such that $\lca(i,j)$ is lowest among all nodes in $S$, as defined above, $x$ denote $\lca(i,j)$, and $P^t_i$ be equal to $\{\beta\}$, where $\beta \neq \alpha$.  From the algorithm it is true that $P^{t+1}_k = \{\alpha\}$.  We let $k \neq i$ be a node in the subtree rooted at $u$.  For any $\gamma \neq \alpha$, $\sigma_k(P^{t+1}_{-k}, \gamma) = \sigma_k(P^{t+1}_{-k}, \gamma)$. The second statement is established for $k$ by the fact that since $P_k^{t+1} = P_k^t \neq \{\alpha\}$, we have $\mu_k(P^{t+1}) = \mu_k(P^t)$.  Similarly for node $i$, we have $\mu_i(P^{t+1}) = (\alpha,v) \ultraPairPrefStrict{i} \mu_i(P^t)$.

To establish the first statement or any node $k$ in the subtree
rooted at $u$ we consider two cases.  Let $y$ be the child of $x$ that is an ancestor of $j$ (see Figure~\ref{fig:examples}).  In the first case, we let $k$ be in the subtree rooted at $y$.  Then, by our choice of $i$, it is true that
\[
\mu_k(P^{t+1}) \ultraPairPref{k} (\alpha, v) \ultraPairPref{k} (\alpha, x) = (\alpha, \sigma_k(P^{t+1}_{-k}, \alpha)),
\]
which, by Lemma~\ref{lem:sum_pair_pref}, implies that the first statement holds for $k$.  In the second case, we let $k$ be in the subtree
rooted at $u$, but not in the subtree rooted at $y$.  Again,
$\sigma_k(P^{t+1}_{-k}, \gamma) = \sigma_k(P^t_{-k}, \gamma)$ for
$\gamma \neq \alpha$.  And for $\alpha$ it is true that
\[
(\alpha, \lca(k,\sigma_k(P^{t+1}_{-k}, \alpha))) = (\alpha, \lca(k,i)) \ultraPairPref{k} (\alpha, x) \ultraPairPref{k} \mu_k(P^t) = \mu_k(P^{t+1}),
\]
which establishes the first statement for $k$ using Lemma~\ref{lem:sum_pair_pref}.

To establish the third statement we use the fact that $|W_{t+1}|
\le |W_t|$, which is immediate from the definition of the algorithm's $t$ step.  When $S$ is empty, $|W_{t+1}| < |W_t|$ since a fictional
player is deleted.  When $S$ is nonempty, we have proved above that $\mu_i(P^{t+1}) \ultraPairPrefStrict{i} \mu_i(P^t)$.  This concludes the proof of the third statement and of the whole lemma. 
\end{proof}

\begin{theorem} \label{theo:ultra}
For hierarchical node preferences, an equilibrium can be found in
polynomial time.
\end{theorem}

\begin{proof}
From Lemma~\ref{lem:ultra} and the definition of the algorithm it is straightforward that it returns a valid equilibrium at the termination.  We should prove now that the termination is achieved in polynomial time.  We consider the potential function given by the sum of $|W_t|$ and the sum of the $\mu_i(P^t)$ position in the preorder $\ultraPairPref{i}$ over all $i$.  We notice that $|W_0|$ is at most $nm$, where $m$ is the number of objects and $n$ is $|\nodeSet|$ which is at least the number of internal nodes.  Moreover, the initial potential is at most $nm + n^2m$ since $|\objSet \times \internalSet|$ is at most $nm$.  From Lemma~\ref{lem:ultra}, the potential decreases by at least one in each step and thus the number of steps is at most $nm + n^2m$.

We need to also prove that each step can be implemented in polynomial time.  In the initialization we add $O(nm)$ fictional players and compute the best response for each node $i \in V$.  For the later process, we compare at most $m$ placements for each $k \in V$, i.e. one for each object.  During each subsequent step we select a fictional player $j$, we determine whether the set $S$ is nonempty, and if so we compute node $i$ and the updated placement.  From this process we only need to explain the computation of $S$ and $i$, where $S$ is the set of all $k$ nodes which are not in equilibrium when a fictional player $j$ is deleted.  $S$ is computed as follows: for each node $k \in V$, we replace the current object in its cache by $\alpha$ and add $k$ to $S$.  According to the utility, this yields to a more preferable placement.  Thus, $S$ can be computed in time polynomial in $n$.  To complete the proof of the theorem, we let node $i$ simply be a node in $S$ such that $\lca(i,j)$ is lowest among all nodes in $S$.  This can also be computed in time polynomial in $n$.
\end{proof}

\section{A general framework for \cc\ games with ordinal preferences}\label{sec:model}
In this section, we present a new framework on \cc\ games with ordinal preferences, to generalize the results that were presented in Section~\ref{sec:sum_ultrametric} to a broad class of utility functions, and to also enable the study of the existence and complexity of equilibria in more general settings.

\smallskip
\BfPara{Node preference relations}  Among all the nodes in $V$, we assume that each node $i \in \nodeSet$ has a total preorder $\nodePref{i}$\footnote{We define a total preorder as a binary relation that satisfies reflexivity, transitivity, and totality.  By totality we mean that for any $i,j,k$, either $j \nodePref{i} k$ or $k \nodePref{i} j$.} and $\nodePref{i}$ further satisfies $i \nodePref{i} j$ for all $i,j \in V$.  A node $i$ {\em prefers}\/ $j$ over $k$ if $j \nodePref{i} k$, while a node $j$ is the {\em most $i$-preferred}\/ in a set $S$ of nodes if $j \in S$ and $j \nodePref{i} k$ for all $k \in S$.  We let $j \nodePrefEq{i} k$ denote that $j \nodePref{i} k$ and $k \nodePref{i} j$, while when it is not the case that $k \nodePref{i} j$, we denote it by $j \nodePrefStrict{i} k$.  Notice that $\nodePrefStrict{i}$ is a strict weak order\footnote{A strict weak order is a strict partial order $>$, i.e. a transitive relation that is irreflexive, in which the ``neither $a > b$ nor $b > a$''  relation is transitive.  Strict weak orders and total preorders are widely used in the field of microeconomics.} and for any $i, j, k \in \nodeSet$ exactly one of the following three relations hold: 1) $j \nodePrefStrict{i} k$, 2) $k \nodePrefStrict{i} j$, and 3) $k \nodePrefEq{i} j$.  We also extend the $\sigma_i(P, \alpha)$ and $\sigma_i(P_{-i}, \alpha)$ notations such that they denote a most $i$-preferred node holding $\alpha$ in $P$ and $P_{-i}$ respectively, breaking ties arbitrarily.

The access cost function $\costNoArgs$ introduced in Section~\ref{sec:sum_model} induces a natural node preference relation: $j \nodePrefStrict{i} k$ if $\cost{i}{j} < \cost{i}{k}$, and $j \nodePrefEq{i} k$ if $\cost{i}{j} = \cost{i}{k}$.  In fact, as we show in Lemma~\ref{lem:acyclic}, undirected networks (i.e., when the access cost function is symmetric) are equivalent to acyclic node preference collections.  Formally, the collection $\{\nodePref{i}: i \in V\}$ is an {\em acyclic node preference collection}\/ if there does not exist a sequence of nodes $i_0, i_1, \ldots, i_{k-1}$ for an integer $k \ge 3$ such that $i_{(j-1) \bmod k} \nodePrefStrict{i_j} i_{(j+1) \bmod k}$ for all $0 \le j < k$.

\begin{lemma}
\label{lem:acyclic}
Any undirected network yields an acyclic node preference collection.  For any acyclic node preference collection, we can compute, in polynomial time, symmetric cost functions that are consistent with the node preferences.
\end{lemma}

\begin{proof}
Let $\costNoArgs$ denote a symmetric access cost function over the set $V$ of nodes.  For a given node $i \in V$, we have $j \nodePref{i} k$ iff $\cost{i}{j} \le \cost{i}{k}$.  We now argue   that the collection $\{\nodePref{i}: i \in V\}$ is acyclic.  Suppose, for the sake of contradiction, that there exists a sequence of nodes $i_0, i_1, \ldots, i_{k-1}$ for an integer $k \ge 3$ such that $i_{(j-1) \bmod k} \nodePrefStrict{i_j} i_{(j+1) \bmod k}$ for all $0 \le j < k$.  It then follows that:
\[
\cost{i_j}{i_{(j-1) \bmod k}} < \cost{i_j}{i_{(j + 1) \bmod k}} \mbox{ for } 0 \le j < k.
\]
Since $\costNoArgs$ is symmetric, we obtain
\[
\cost{i_j}{i_{(j-1) \bmod k}} < \cost{i_{(j+1) \bmod k}}{i_j} \mbox{ for } 0 \le j < k,
\]
which is a contradiction, since $\cost{i_0}{i_{(k-1)}} < \cost{i_1}{i_0} < \dots < \cost{i_{(k-1)}}{i_0} = \cost{i_0}{i_{(k-1)}}$.

Given an acyclic collection of node preferences, we compute an associated access cost function $\costNoArgs$ in polynomial time as follows.  We construct a directed graph $G$ over the set $U$ of all unordered pairs $(i,j): i,j \in V$, $i \neq j$. There is a directed edge from node $(i,j)$ to $(i,k)$ if and only if $k \nodePref{i} j$.
Since the collection $\{\nodePref{i}: i \in V\}$ is acyclic, $G$ is a dag.  We compute the topological ordering $\pi: U \rightarrow \mathbb{Z}$; thus, we have $\pi((i,j)) < \pi((k, \ell))$ whenever there is a directed path from $(i,j)$ to $(k,\ell)$.  Setting $\cost{i}{j}$ to be $\pi((i,j))$ gives us the desired undirected network.
\end{proof}

\smallskip \BfPara{Utility preference relations}  Each node in our game-theoretic model attaches a utility to each global placement.  In our general definition a large class of utility functions it is considered simultaneously.  Instead of defining a numerical utility function, we let the utility at each node $i$ be a total preorder $\placePref{i}$ among the set of all global placements.  The $\placePrefStrict{i}$ and $\placePrefEq{i}$ notations over global placements are defined analogously.   We require that $\placePref{i}$,
for each $i \in V$, satisfies the following two basic conditions:

\begin{itemize}
  \item {\bf Monotonicity}: If for any two global placements $P$ and $Q$, for each object $\alpha$, and each node $q$ with $\alpha \in Q_q$, there exists a node $p$ with $\alpha \in P_p$ and $p \nodePref{i} q$, then $P \placePref{i} Q$.

  \item {\bf Consistency}: Let two global placements $(P_i, P_{-i})$ and $(Q_i, Q_{-i})$ such that for each object $\alpha \in P_i \cup Q_i$, if $p$ (resp. $q$) is a most $i$-preferred node in $\nodeSet \setminus \{i\}$ holding $\alpha$, i.e. $\alpha \in P_p$ (resp. $\alpha \in Q_q$), then $p \nodePrefEq{i} q$.  If $(P_i, P_{-i}) \placePrefStrict{i} (Q_i, P_{-i})$, then $(P_i, Q_{-i}) \placePref{i} (Q_i, Q_{-i})$.
\end{itemize}

In words, the monotonicity condition says that for any node, if all the objects in a placement are placed at nodes that are at least as preferred as in another placement, then the node prefers the former placement at least as much as the latter.  The consistency condition says that the preference for a node to store one set of objects instead of another is entirely a function of the set of most preferred other nodes that together hold these objects.  For instance, if a node $i$ with unit capacity prefers to store $\alpha$ over $\beta$ in a scenario where the most $i$-preferred node (other than $i$) storing $\alpha$ (resp. $\beta$) is $j$ (resp. $k$), then $i$ prefers to store $\alpha$ at least as much as $\beta$ in any other situation where the most $i$-preferred node (other than $i$) storing $\alpha$ (resp. $\beta$) is $j$ (resp. $k$).

\smallskip \BfPara{Generality of the conditions} We note that many standard utility functions defined for replica placement problems~\cite{ChunCWBPK04,LaoutarisSOSB07,PollatosSelfish}, including the sum and max functions, satisfy the monotonicity and consistency conditions.  Indeed, any utility function that is an $L_p$ norm, for any $p$, over the costs for the individual objects, also satisfies the conditions.  Furthermore, since the monotonicity and consistency conditions apply to the individual utility functions, our model allows the different nodes to adopt different types of utilities, as long as each separately satisfies the two conditions.

\smallskip \BfPara{Binary object preferences}  One of the utility preference relations classes we study is based on binary object preferences.  Assume that each node $i$ is equally interested in an objects set $S_i$ and it does not have any interest in the other objects.  Then, $\tau_i(P)$ will denote the $|S_i|$-length sequence of the $\sigma_i(P,\alpha)$, such that $\alpha \in S_i$ and it is in non-increasing order based on the $\nodePref{i}$ relation.  In this setup the consistency condition can be further strengthened to the {\bf binary consistency} term:  for any placements $P = (P_i, P_{-i})$ and $Q = (Q_i, Q_{-i})$ with $P_{-i} = Q_{-i}$, we let $P \placePref{i} Q$ if and only if for $1 \le k \le |S_i|$, the $k^{th}$ component of $\tau_i(P)$ is at least as $i$-preferred as the $k^{th}$ component of $\tau_i(Q)$.

\smallskip \BfPara{\icc\ Games} We let a \icc\ game be a tuple $(\nodeSet, \objSet, \{\nodePref{i}\}, \{\placePref{i}\})$ in the general axiomatic framework.  A pure Nash equilibrium in a \icc\ game instance is a global placement $P$ such that there is no placement $Q_i$ for which $(Q_i, P_{-i}) \placePrefStrict{i} (P_i, P_{-i})$, for each $i \in \nodeSet$.

To further analyse the complexity results, a definition of a game instance specification is required.  We first specify the set $\nodeSet$, the node cache capacities, and an enumerated list of object names $\objSet$.  For each node $i \in \nodeSet$, we specify $i$'s preference relation $\nodePref{i}$ succinctly by a set of at most $\binom{n}{2}$ bits.  However, the utility preference relation $\placePref{i}$ is over a potentially exponential number of placements in terms of $n$, $m$, and cache sizes.  We further assume that the utility preference relations are specified by an efficient algorithm, which we denote as {\em utility preference oracle}\/, that takes as input a node $i$, and two global placements $P$ and $Q$, and returns whether $P \placePref{i} Q$.  For the sum, max, and $L_p$-norm utilities, the utility preference oracle simply computes the relevant utility function.  For binary object preferences, the binary consistency condition yields an oracle which is polynomial in the number of nodes, objects, and cache sizes.

\smallskip
\BfPara{Unit cache capacity} We now argue that the unit cache capacity assumption of Section~\ref{sec:sum_model} continues to hold without loss of generality.  Consider a set $\nodeSet$ of nodes in which the cache of node $i$ can store $c_i$ objects.  Let $\nodeSet'$ denote a new set of nodes which contains, for each node $i$ in $\nodeSet$, new nodes $i_1, i_2, \ldots, i_{c_i}$, i.e., one new node for each unit of the cache capacity of $i$.  We set the node preferences as follows: for all $i, i', j \in \nodeSet$, $1 \le f, \ell \le c_j$, $1 \le k, k' \le c_i$, we have $i_{k} \nodePref{j_\ell} i'_{k'}$ whenever $i \nodePref{j} i'$, and $j_f \nodePrefEq{i_k} j_\ell$.

We consider an obvious onto mapping $f$ from placements in $\nodeSet'$ to those in $\nodeSet$.  Given placement $P'$ for $\nodeSet'$, we set $f(P') = P$ where $P_i = \cup_{1 \le k \le c_i} P'_{i_k}$.  This mapping naturally defines the utility preference relations for the node set $\nodeSet'$.  In particular, for any $i \in \nodeSet$ and $1 \le k \le c_i$, $P' \placePref{i_k} Q'$ whenever $f(P') \placePref{i} f(Q')$.  We also note that $f$ is computable in time polynomial in the number of nodes and the sum of the cache capacities.  It is easy to verify that the utility preference relation $\placePref{i_k}$ for all $i_k \in \nodeSet'$ satisfies the monotonicity and consistency conditions.  Furthermore, $P'$ is an equilibrium for $\nodeSet'$ if and only if $f(P')$ is an equilibrium for $\nodeSet$; this together with the onto property of the mapping $f$ gives us the desired reduction.

\section{Existence of equilibria in the general framework} \label{sec:exist}
In this section, we establish the existence of equilibria for several \icc\ games under the general framework of \cc\ games with ordinal preferences that we introduced in Section~\ref{sec:model}.  First, we extend the sum utility function results on hierarchical networks to the general framework (Section~\ref{sec:ultrametric}).  Next, we show that \icc\ games on undirected networks and binary object preferences are potential games (Section~\ref{sec:01-metric}).  Finally, when there are only two objects in the system, we use the technique of fictional players to give a polynomial-time construction of equilibria for \icc\ games on undirected networks (Section~\ref{sec:2-obj-metric}).

\subsection{Hierarchical networks} \label{sec:ultrametric}
We fist show that the polynomial time algorithm which was introduced in Section~\ref{sec:sum_ultrametric} holds also for the general framework of \cc\ games with ordinal preferences.  A hierarchical network, as defined in the general framework, is a tree $T$ whose  leaves set is the node set $\nodeSet$ and the node preference relation $\nodePref{i}$ is $j \nodePref{i} k$ if $\lca(i,j)$ is a descendant of $\lca(i,k)$.  This hierarchical network structure and each node's $i$ pair-preference relations $\ultraPairPref{i}$, determine completely the analysis of the algorithm introduced in Section~\ref{sec:sum_ultrametric}.  The latter were defined for the sum utility function.  Extending our analysis to the general framework, requires a new preference relation derivation and the establishment of Lemma's~\ref{lem:sum_pair_pref} analogue, which we present next for arbitrary utility preference relations that satisfy the monotonicity and consistency properties.

\smallskip \BfPara{Pair preference relations} For any utility preference relation $\placePref{i}$ that satisfies the monotonicity and consistency conditions, we define a strict weak order $\pairPrefStrict{i}$ on $\objSet \times \ancestors{i}$, where $\ancestors{i}$ is the set of $i$'s proper ancestors in $T$.

\begin{enumerate}
  \item \label{pair:1} We let $(\alpha, v) \pairPrefStrict{i} (\alpha, w)$ hold whenever $v$ is a proper ancestor of $w$, for each object $\alpha$, node $i$, and proper $i$'s ancestors $v$ and $w$.

  \item \label{pair:2} Considering distinct objects $\alpha, \beta$ and nodes $i,j,k$ with $j,k \neq i$, we let ${\cal P}$ be the set of global placements $P$, such that $j$ (resp. $k$) is a most $i$-preferred node in $V \setminus \{i\}$ holding $\alpha$ (resp. $\beta$) in $P_{-i}$.  If {\em there exist}\/ global placements $P = (\{\alpha\}, P_{-i})$ and $Q = (\{\beta\}, P_{-i})$ in ${\cal P}$ with $P \placePrefStrict{i} Q$, then $(\alpha, \lca(i,j)) \pairPrefStrict{i} (\beta, \lca(i,k))$.  
\end{enumerate}

In words, item~\ref{pair:1} says that $i$'s preference for keeping
$\alpha$ in its cache increases as the most $i$-preferred node holding
$\alpha$ becomes less preferred (or ``moves farther away'').  In
item~\ref{pair:2}, $(\alpha, v) \pairPrefStrict{i} (\beta, w)$ means
that if $i$ needs to place either $\alpha$ or $\beta$ in its cache,
and the least common ancestor of $i$ and the most $i$-preferred node
in $\nodeSet \setminus \{i\}$ holding $\alpha$ (resp., $\beta$) is $v$
(resp., $w$), then $i$ prefers to store $\alpha$ over $\beta$.  The
strict weak order $\pairPrefStrict{i}$ induces a total preorder
$\pairPref{i}$ as follows: $(\alpha, v) \pairPref{i} (\beta, w)$ if it
is not the case that $(\beta, v) \pairPrefStrict{i} (\alpha, w)$.  We
similarly define $\pairPrefEq{i}$: $(\alpha, v) \pairPrefEq{i} (\beta,
w)$ if $(\alpha, v) \pairPref{i} (\beta, w)$ and $(\beta, v)
\pairPrefEq{i} (\alpha, w)$.

\begin{lemma}
\label{lem:strict_weak_order}
  For each $i$, $\pairPrefStrict{i}$ as given above, is a well-defined
  strict weak order.
\end{lemma}

\begin{proof}
We need to ensure the well-definedness of part~\ref{pair:2} of the
definition of pair preference relations.  That is, we need to show
that for any placements $P_{-i}$ and $Q_{-i}$ such that a most
$i$-preferred node in $P_{-i}$ holding $\alpha$ (resp., $\beta$) is
also a most $i$-preferred node in $Q_{-i}$, it is impossible that
$(\{\alpha\}, P_{-i}) \placePrefStrict{i} (\{\beta\}, P_{-i})$ and
$(\{\beta\}, Q_{-i}) \placePrefStrict{i} (\{\alpha\}, Q_{-i})$ both
hold.  This directly follows from the consistency condition for
utility preference relations.

The reflexivity and transitivity of $\pairPref{i}$ are immediate from the definitions and the reflexivity and transitivity of $\placePref{i}$.  Finally, to ensure the well-definedness of the strict preorder $\pairPrefStrict{i}$, we also have to show that there is no collection of pairs $(\alpha_j, v_j)$, $0 \le j < \ell$ for some integer $\ell > 1$, such that $(\alpha_j, v_j) \pairPrefStrict{i} (\alpha_{j+1 \bmod \ell}, v_{j+1 \bmod \ell})$ for $0 \le j < \ell$.  To see this, it is sufficient to note that if $(\alpha, v) \pairPrefStrict{i} (\alpha', v')$ then for all placements $P$ and $P'$ such that $P_{-i} = P'_{-i}$ and the least common ancestor of $i$ and the most $i$-preferred node in $V \setminus \{i\}$ that holds $\alpha$ (resp. $\alpha'$) is $v$ (resp. $v'$) we have $P \placePrefStrict{i} P'$.  So any cycle in the strict preorder $\pairPrefStrict{i}$ implies a cycle in $\placePrefStrict{i}$, yielding a contradiction.
\end{proof}

Analogous to Lemma~\ref{lem:sum_pair_pref}, we can express the best response of any player in hierarchical networks as follows.  For any global placement $P = (\{\alpha\}, P_{-i})$, assume that $j$ (resp. $k$) is a most $i$-preferred node holding object $\alpha$ (resp. $\beta$) in $P_{-i}$, and $(\{\beta\}, P_{-i}) \placePrefStrict{i} (\{\alpha\}, P_{-i})$, i.e., for node $i$, storing $\beta$ is a better response to $P_{-i}$ than storing $\alpha$.  Then the following Lemma holds.

\begin{lemma} \label{lem:pair_pref}
For any global placement $P = (\{\alpha\}, P_{-i})$, $(\beta, \lca(i,k)) \pairPrefStrict{i} (\alpha, \lca(i,j))$.  Furthermore, $\{\alpha\}$ is a best response to $P_{-i}$, where $\alpha$ maximizes $(\gamma, \lca(i,\sigma_i(P_{-i}, \gamma)))$, over all objects $\gamma$, according to $\pairPref{i}$.
\end{lemma}

\begin{proof}
The first statement of the lemma directly follows from item 2 of the definition of pair preference relations.  We establish the second statement by contradiction.  Suppose that for node $i$, $\{\beta\}$ is a better response to $P_{-i}$ than $\{\alpha\}$.  Then, we have $(\{\beta\}, P_{-i}) \placePrefStrict{i} (\{\alpha\}, P_{-i})$, which, by item 2 of the definition of pair preference relations, implies that $(\beta, \lca(i,\sigma_i(P_{-i}, \beta))) \pairPrefStrict{i} (\alpha, \lca(i,\sigma_i(P_{-i}, \alpha)))$, a contradiction to the choice of $\alpha$.
\end{proof}

The remainder of the analysis for hierarchical networks (Lemma~\ref{lem:ultra} and Theorem~\ref{theo:ultra}) follows as before, invoking Lemma~\ref{lem:pair_pref} instead of Lemma~\ref{lem:sum_pair_pref}.

\subsection{Undirected networks with binary object preferences} \label{sec:01-metric}
Let $\costNoArgs$ be a symmetric cost function for an undirected network over the node set $\nodeSet$.  From the binary object preferences definition for each node $i$ we are given an object set $S_i$ in which $i$ is equally interested.  We prove the existence of equilibria via a potential function argument.  Given a placement $P$, we let $\Phi_i(P) = \cost{i}{j}$, where $j$ is the most $i$-preferred node in $V-\{i\}$ holding the object in $P_i$.  We introduce the potential function $\Phi$: $\Phi(P) = (\Phi_0, \Phi_{i_1}(P), \Phi_{i_2}(P), \ldots, \Phi_{i_n}(P))$, where $\Phi_0$ is the number of nodes $i$ such that $P_i \subseteq S_i$, and $\Phi_{i_j}(P) \leq \Phi_{i_{j+1}}(P)$, $\forall j$, where $V=\{i_1, i_2, \ldots, i_n\}$.  We prove that $\Phi$ is an increasing potential function, i.e. after any better response step, $\Phi$ increases in lexicographical order.

Let $P =(P_i, P_{-i})$ be an arbitrary global placement. Assume that $P_i = \{\alpha\}$ and $j$ is the most $i$-preferred node in $P_{-i}$ holding $\alpha$.  Consider any better response step, from placement $P$ to $Q = (Q_i, P_{-i})$, where $Q_i = \{\beta\}$.  Clearly $\beta \in S_i$.  We consider two cases.  First, suppose $\alpha \notin S_i$ and $\beta \in S_i$. Then, $\Phi_0$ increases, and so does the potential.  The second case is where $\alpha, \beta \in S_i$.  Let $k$ be the most $i$-preferred node in $P_{-i}$ holding $\beta$.  In this case, $\Phi_0$ does not change.  However, since this is a better response step of $i$, $j \nodePrefStrict{i} k$, implying that $\cost{i}{k} > \cost{i}{j}$ and hence $\Phi_i(Q) > \Phi_i(P)$.  Consider any other node $j$.  If $j$ holds any object $\gamma$ other than $\beta$, since no new copy of $\gamma$ has been added, $\Phi_j(Q) \ge \Phi_j(P)$.  It remains to consider the case where $j$ holds $\beta$.  If $S$ is the set of nodes in $\nodeSet \setminus \{j\}$ holding $\beta$ in $P_{-j}$, then $S \cup \{i\}$ is the set of nodes in $\nodeSet \setminus \{j\}$ holding $\beta$.  Thus, $\Phi_j(Q) = \min\{\Phi_j(P), \cost{j}{i}\} \ge \min\{\Phi_j(P), \Phi_i(Q)\}$.  This also means that $\Phi_j(P)$ appears later in the sorted order than $\Phi_i(P)$ and $\Phi_j(Q)$ appears no earlier in the sorted order than $\Phi_i(Q)$.  Hence, $\Phi(Q)$ is lexicographically greater than $\Phi(P)$.  This establishes that for undirected networks with binary object preferences, the resulting \icc\ game is a potential game, and hence also in \pls~\cite{JPY88}.

\subsection{Undirected networks with two objects} \label{sec:2-obj-metric}
In the case of an undirected network with two objects we provide a polynomial-time algorithm to compute an equilibrium.  We use the fictional player technique that was introduced in Section~\ref{sec:ultrametric}.  In the beginning a set of fictional players are introduced to serve the two objects in the network at zero cost from each node.  In each subsequent step, the fictional players are progressively moved ``further'' away, in a way that at each instance the equilibrium is ensured.  The whole set of fictional players are completely removed when they are at the least preferred cost from all the nodes, yielding finally to an equilibrium for the original network.

Suppose we are given a undirected network with access cost function $\costNoArgs$.  Also let $\DstList$ be the set $\{ 0, \ell_1, \ell_2, \ldots, \ell_r \}$ of all access costs between nodes in the system in increasing order; that is, $\ell_1 = \min_{i,j} \cost{i}{j}$ and $\ell_r = \max_{i,j} \cost{i}{j}$ and $\ell_i < \ell_{i+1}$ for all $1 \le i < r$. 

\smallskip \BfPara{Fictional player} For an object $\alpha$, a fictional $\alpha$-player is a new node that will store $\alpha$ in every equilibrium; an fictional $\alpha$-player prefers storing $\alpha$ over any other object.  We denote by $\auxsrv{\ell}{\alpha}$ the fictional $\alpha$-player which is at access cost $\ell$ from every node in $V$.

\smallskip \BfPara{The algorithm}

\smallskip \ItPara{Initialization} Assuming that there are two objects $\alpha$ and $\beta$ in the system, we initially set up a fictional $\alpha$-player $\auxsrv{0}{\alpha}$ and $\beta$-player $\auxsrv{0}{\beta}$ at access cost $0$ from each node in $V$, which does not affect the actual distance between nodes.  We let nodes replicate their most preferred object and access the other without any access cost from the corresponding fictional player.  This placement is obviously an equilibrium.

\smallskip \ItPara{Step $t$ of algorithm} Fix an equilibrium $P$ for the node set $V \cup \{\auxsrv{\ell_t}{\alpha}\} \cup \{\auxsrv{\ell_t}{\beta}\}$. We describe one step of the algorithm which computes a new set of fictional players $\auxsrv{\ell_{t+1}}{\alpha}$ and $\auxsrv{\ell_{t+1}}{\beta}$ and a new placement $P'$ such that $P'$ is an equilibrium for the node set $V \cup \{\auxsrv{\ell_{t+1}}{\alpha}\} \cup \{\auxsrv{\ell_{t+1}}{\beta}\}$.  We first remove the $\alpha$-player $\auxsrv{\ell_{t}}{\alpha}$ from the system and instead we add $\auxsrv{\ell_{t+1}}{\alpha}$.  If there do not exist nodes that want to deviate we are done. Otherwise, assume that there exists a node $i$ that wants to deviate from its strategy. Since the most $i$-preferred node holding $\beta$ in $V \cup \{\auxsrv{\ell_{t}}{\alpha}\} \cup \{\auxsrv{\ell_{t}}{\beta}\}$ remains the same in $V \cup \auxsrv{\ell_{t+1}}{\alpha} \cup \auxsrv{\ell_{t}}{\beta}$, $i$ is not holding object $\alpha$.  Thus the only nodes that may want to deviate are those that are holding object $\beta$. We argue that if we let $i$ to deviate from $\beta \in P_i$ to $\alpha \in P'_i$, there is no node $j \in V \setminus \{i\}$ that gets affected by $i$'s deviation.  Consider the following two cases:

\begin{itemize}
	\item If a node $j$ has access cost at most $\ell_{t}$ from $i$, then $\beta \in P_j$. Otherwise, if $\alpha \in P_j$, $\auxsrv{\ell_{t}}{\alpha}$ would not be the most $i$-preferred node holding $\alpha$ and thus $i$ would not be affected by any change of $\alpha$-players.  Thus there does not exist any node $j \in V \setminus \{i\}$ with access cost at most $\ell_{t}$ from $i$, such that $\alpha \in P_j$, and as we showed above $\alpha \in P'_j$.
		
	\item If a node $j$ has access cost at least $\ell_{t+1}$ from $i$, then $P_j = P'_j$.  Because of the $\alpha$-player $\auxsrv{\ell_{t+1}}{\alpha}$ and the $\beta$-player $\auxsrv{\ell_{t}}{\beta}$, $i$ would never be the $j$-most preferred node in $P'$.
\end{itemize}

We then remove the $\beta$-player $\auxsrv{\ell_{t}}{\beta}$ from the system and instead we add $\auxsrv{\ell_{t+1}}{\beta}$.  Using a similar argument as above, we obtain a new equilibrium at the end of this step.
	
\begin{theorem}
	For undirected networks with two objects, an equilibrium can be found in polynomial time.
\end{theorem}

\begin{proof}
An initial placement $P$, where we have the set of fictional players $\auxsrv{0}{\alpha}$ and $\auxsrv{0}{\beta}$ in the system, is obviously an equilibrium.  It is immediate from our argument above that at termination the algorithm returns a valid equilibrium.

The size of the set $\DstList$ is at most ${n \choose 2}$ which is at most $n^2$.  In each step $t$ at most $n$ nodes may want to deviate from their strategy, since we showed above that if a node deviates once in a step, it will not deviate again during the same step. Thus, the total number of deviations in the algorithm is at most $n^3$. 
\end{proof}

\section{Non-Existence of equilibria in \icc\ games and the associated decision problem}
\label{sec:non-exist}
In this section, we show that the classes of games studied in Section~\ref{sec:exist} are essentially the only games where equilibria are guaranteed to exist.  We identify the most basic \icc\ games where equilibria may not exist, and study the complexity of the associated decision problem.

\subsection{NP-Completeness}
\label{sec:npc}
We first show that  it is \NP-hard to determine whether a given \icc\ game has an equilibrium even when the utility preference relations are based on the sum utility function and either the number of objects is small or the object preferences are binary.  Some simple network examples appear in Fig.~\ref{Fig:fig2} (middle and right)--the networks are described in details after Theorem~\ref{app:npc-theorem}--showing that there does not exist an equilibrium in these configurations (proved in the second part of Theorem~\ref{Theorem:nonmetric} and~\ref{Theorem:metric}).  The \NP-hardness proof is by a polynomial-time reduction from 3SAT~\cite{GareyJohnson}.  Each reduction is built on top of a gadget which has an equilibrium if and only if a specified node holds a certain object.  Several copies of these gadgets are then put together to capture the given 3SAT formula.

\begin{theorem} \label{app:npc-theorem}
  The problem of determining whether a \icc\ instance has an equilibrium is in \NP\ even if one of these three restrictions hold: (a) the number of objects is two; (b) the object preferences are binary and number of objects is three; (c) the network is undirected and the number of objects is three.
\end{theorem}

The membership in \NP\ is immediate, since one can determine in polynomial time whether a given global placement is an equilibrium.  The remainder of the proof focuses on the hardness reduction from 3SAT.

Given a 3SAT formula $\phi$ with $n$ variables $x_1$, $x_2$, $\ldots$, $x_n$ and $k$ clauses $c_1$, $c_2$, $\ldots$, $c_k$, we construct a \icc\ instance as follows.  For each variable $x_i$ in $\phi$, we introduce two variable nodes $X_i$ and $\bar{X}_i$.  We set $\cost{X_i}{\bar{X}_i}$ and the symmetric $\cost{\bar{X}_i}{X_i}$ to be $0.5$, where $\costNoArgs$ is the underlying access cost function.  For each clause $c_j$ we introduce a clause node $C_j$.  Assuming that $\ell_{j,r}$ for $r \in \{1,2,3\}$, are the three literals of the $c_j$ clause in formula $\phi$, we set $\cost{C_j}{L_{j,r}}$ and $\cost{L_{j,r}}{C_j}$ to be $1$, where $L_{j,r}$ is the corresponding variable node.  We also introduce a gadget $G$ illustrated in Figure~\ref{Fig:fig2} (middle and right), consisting of nodes $\node{S}$, $\node{A}$, $\node{B}$, and $\node{C}$.  We set the access cost $\cost{S}{C_i}$ and the symmetric $\cost{C_i}{S}$, for all $1 \leq i \leq k$ between node $\node{S}$ and all clause nodes to be $2$.  The general construction is illustrated in Figure~\ref{Fig:fig2} (left).

\begin{figure}[ht]
\centering
\includegraphics[scale=0.17]{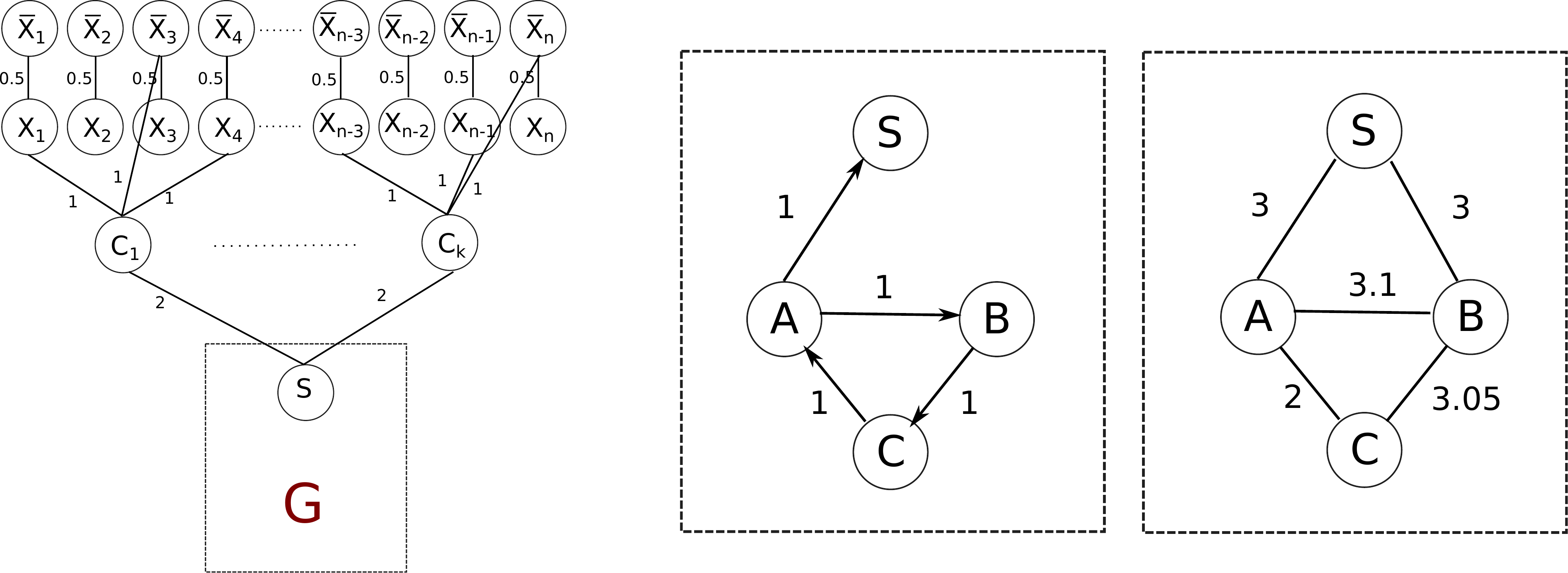}
\caption{{\small Left: instance of the construction for the undirected case proof of NP-Hardness, where $\phi = (x_1 \vee \bar{x_3} \vee x_4) \wedge \ldots \wedge (x_{n-3} \vee x_{n-1} \vee \bar{x_n})$}. {\small Middle and right: gadget G for the directed and the undirected case}}
\label{Fig:fig2}
\end{figure}

\BfPara{Directed networks with two objects} We set the access costs $\cost{A}{S} = \cost{A}{B} = \cost{B}{C} = \cost{C}{A} = 1$, and the server node, which stores a fixed copy of two objects $\object{\alpha}$ and $\object{\beta}$, at access cost $\distsrv = 10$ from all nodes in $\nodeSet$.  We also set the weights of the variable nodes $\rate{x_i}{\alpha} = \rate{\bar{x}_i}{\alpha} = \rate{x_i}{\beta} = \rate{\bar{x}_i}{\beta} = 1$, the weights of the clause nodes $\rate{C_i}{\alpha} = 0.85$ and $\rate{C_i}{\beta} = 1$, for all $1 \leq i \leq k$.  Finally, we set the weights of the nodes in the G gadget $\rate{S}{\alpha} = 0.85$, $\rate{S}{\beta} = \rate{A}{\beta} = \rate{B}{\alpha} = \rate{B}{\beta} = \rate{C}{\alpha} = \rate{C}{\beta} = 1$ and $\rate{A}{\alpha} = 0.7$.  We refer to this \icc\ instance as $I_1$.

\BfPara{Undirected networks with three objects} We set the access costs $\cost{A}{S} = \cost{B}{S} = 3$, $\cost{A}{B} = 3.1$, $\cost{B}{C} = 3.05$, and $\cost{C}{A} = 2$; while symmetry holds. The server node, which stores a fixed copy of three objects
$\object{\alpha}$, $\object{\beta}$, and $\object{\gamma}$, is at access cost $\distsrv = 5$ from all nodes in $\nodeSet$.  We set
the weights of the clause nodes $\rate{x_i}{\alpha} = \rate{\bar{x}_i}{\alpha} = \rate{x_i}{\beta} = \rate{\bar{x}_i}{\beta} = 1$,  the weights of the clause nodes $\rate{C_i}{\alpha} = 0.85$ and $\rate{C_i}{\beta} = 1$, for all $1 \leq i \leq k$.  Finally, we set the weight of the nodes in the G gadget $\rate{S}{\alpha} = 0.85$, $\rate{S}{\beta} = \rate{A}{\alpha} = \rate{B}{\beta} = \rate{C}{\beta} = 1$, $\rate{A}{\gamma} = 2$, $\rate{B}{\gamma} = 0.9837$, and $\rate{C}{\gamma} = 1.6$. All the remaining weights are set to $0$.  We refer to this \icc\ instance as $I_2$.

\begin{lemma}
\label{Lemma:var_node}
A variable node $X_i$ holds object $\object{\alpha}$ (resp.,
$\object{\beta}$) if and only if node $\bar{X}_i$ holds object
$\object{\beta}$ (resp., $\object{\alpha}$).
\end{lemma}

\begin{proof}
The proof is immediate, since $\bar{X}_i$ (resp., $X_i$) is $X_i$'s (resp., $\bar{X}_i$'s) nearest node, and both $X_i$ and $\bar{X_i}$ are interested equally in $\alpha$ and $\beta$.
\end{proof}

\begin{lemma}
\label{Lemma:clause_node}
Clause node $\node{C_i}$ holds object $\object{\alpha}$ if and only if its variable nodes $\node{L_{i,j}}$, for $j \in \{1,2,3\}$ hold object $\object{\beta}$.
\end{lemma}

\begin{proof}
First, assume that $\node{L_{i,j}}$, for $j \in \{1,2,3\}$ hold $\object{\beta}$.  These nodes are $C_i$'s nearest nodes holding $\object{\beta}$.  By Lemma~\ref{Lemma:var_node} we know that nodes $\bar{L}_{i,j}$, for $j \in \{1,2,3\}$ hold $\alpha$, and they are $C_i$'s nearest nodes holding $\object{\alpha}$.  Node's $C_i$ cost for holding $\object{\alpha}$ and accessing $\object{\beta}$ from $\node{L_{i,j}}$, for $j\in\{1,2,3\}$, is $\rate{C_i}{\beta} \cost{C_i}{L_{i,j}} = 1$; while the cost for holding $\object{\beta}$ and accessing $\object{\alpha}$ from $\node{\bar{L}_{i,j}}$, for $j \in \{1,2,3\}$, is $\rate{C_i}{\alpha} \cost{C_i}{\bar{L}_{ij}} = 1.275$.  Obviously, node $\node{C_i}$ prefers to replicate $\object{\alpha}$.

Now assume that at least one of the nodes $\bar{L}_{i,j}$, for $j \in \{1,2,3\}$ holds $\object{\alpha}$.  These nodes are $C_i$'s nearest nodes holding $\object{\alpha}$.  Also, by Lemma~\ref{Lemma:var_node}, $C_i$'s nearest nodes holding $\object{\beta}$ are all the remaining nodes from the set $L_{i,j}$, $\bar{L_{i,j}}$, for $j \in \{1,2,3\}$, that don't hold $\alpha$.  Node's $C_i$ cost for holding $\object{\beta}$ and accessing $\object{\alpha}$ from $\node{L_{i,j}}$, for $j \in \{1,2,3\}$, is $\rate{C_i}{\alpha} \cost{C_i}{L_{i,j}} = 0.85$; while the cost for holding $\object{\alpha}$ and accessing $\object{\beta}$ from node $\node{\bar{L}_{i,j}}$ (resp., $\node{L_{i,j}}$), is $\rate{C_i}{\beta} \cost{C_i}{\bar{L}_{i,j}} = 1.5$ (resp., $\rate{C_i}{\beta} \cost{C_i}{L_{i,j}} = 1$).  Obviously, in any case node $\node{C_i}$ prefers to replicate $\object{\beta}$.
\end{proof}

\begin{lemma}
\label{Lemma:S_node}
Node $\node{S}$ holds object $\object{\alpha}$ if and only if all
clause nodes $\node{C_1}, \ldots, \node{C_k}$ hold object
$\object{\beta}$.
\end{lemma}

\begin{proof}
First, assume that $\node{C_1}, \ldots, \node{C_k}$ are holding $\object{\beta}$.  These nodes are $S$'s nearest nodes holding $\object{\beta}$. Also by Lemma~\ref{Lemma:clause_node}, $S$'s nearest node holding $\object{\alpha}$ is at least one of $\node{L_{i,j}}$ nodes, where $i \in [1,k], j \in \{1,2,3\}$.  The cost for $S$ holding $\object{\alpha}$ and accessing $\object{\beta}$ from a node $\node{C_i}, i \in [1,k]$, is $\rate{S}{\beta} \cost{S}{C_i} = 2$; while the cost for holding $\object{\beta}$ and accessing $\object{\alpha}$ from $\node{L_{i,j}}$, where $i \in [1,k], j \in \{1,2,3\}$, is $\rate{S}{\alpha} \cost{S}{L_{i,j}} = 2.55$.  Obviously, node $\node{S}$ prefers to replicate $\object{\alpha}$.
	
Now assume that at least one of $\node{C_1}, \ldots, \node{C_k}$ holds $\object{\alpha}$.  These nodes are $S$'s nearest node holding $\object{\alpha}$.  Also $S$'s nearest node holding $\object{\beta}$, due to Lemma~\ref{Lemma:clause_node} is one of $\node{L_{i,j}}$, where $i \in [1,k], j \in \{1,2,3\}$.  The cost for holding $\object{\beta}$ and accessing $\object{\alpha}$ from a node $\node{C_i}$, is $\rate{S}{\alpha} \cost{S}{C_i} = 1.7$; while the cost for holding $\object{\alpha}$ and accessing $\object{\beta}$ from a node $\node{L_{ij}}$, where $j \in \{1,2,3\}$, is $\rate{S}{\beta} \cost{S}{L_{i,j}} = 3$.  Obviously, in any case node $\node{S}$ prefers to replicate $\object{\beta}$.
\end{proof}

\begin{theorem}
\label{Theorem:nonmetric}
The \icc\ instance $I_1$ has an equilibrium if and only if node $\node{S}$ holds object $\object{\alpha}$.
\end{theorem}

\begin{proof}
	First, assume that $S$ is holding $\alpha$. By Lemma~\ref{Lemma:S_node} nodes $C_1, \ldots, C_k$ hold object $\beta$, and by Lemma~\ref{Lemma:clause_node} at least one of nodes $\node{L_{i,j}}$, for $j \in \{1,2,3\}$ for each node $C_i, i \in [1,k]$, holds object $\alpha$, and the corresponding $\bar{L}_{i,j}$ is holding object $\beta$. We claim that the placement where $A$ holds $\beta$, $B$ holds $\beta$, and $C$ holds $\alpha$, is a pure Nash equilibrium. We prove this by showing that none of these nodes wants to deviate from their strategy.
	
	Node $A$ does not want to deviate since its cost for holding object $\beta$ and accessing $\alpha$ from $A$'s nearest node $S$, is $\rate{A}{\alpha} \cost{A}{S} = 0.7$; while the cost for holding object $\alpha$ and accessing $\beta$ from $A$'s nearest node $B$, is $\rate{A}{\alpha} \cost{A}{B} = 1$.
	Node $B$ does not want to deviate since its cost for holding object $\beta$ and accessing $\alpha$ from $B$'s nearest node $C$, is $\rate{B}{\alpha} \cost{B}{C} = 1$; while the cost for holding object $\alpha$ and accessing $\beta$ from $B$'s nearest node $A$, is $\rate{B}{\beta} \cost{B}{A} = 2$.
	Node $C$ does not want to deviate since its cost for holding object $\alpha$ and accessing $\beta$ from $C$'s nearest node $A$, is $\rate{C}{\beta} \cost{C}{A} = 1$; while the cost for holding object $\beta$ and accessing $\alpha$ from $C$'s nearest node $S$, is $\rate{C}{\alpha} \cost{C}{S} = 2$.
	Also note that none of $S, C_1, \ldots, C_k, L_{ij}, \bar{L}_{ij}$ for $i \in [1,k], j \in \{1,2,3\}$ is getting affected of the objects been held by the gadget nodes.\\
	
	Now assume that node $\node{S}$ holds object $\object{\beta}$. We are going to prove that for every possible placement over nodes $A$, $B$, and $C$, at least one node wants to deviate from its strategy. Consider the following cases:
	\begin{itemize}
		\item Nodes $A$, $B$, and $C$ hold object $\alpha$: Node $B$ (resp., $C$) wants to deviate, since the cost for holding object $\alpha$ and accessing $\beta$ from $B$'s (resp., $C$'s) nearest node $S$, is $\rate{B}{\beta} \cost{B}{S} = 3$ (resp., $\rate{C}{\beta} \cost{C}{S} = 2$); while the cost for holding object $\beta$ and accessing $\alpha$ from $B$'s nearest node $A$, is $\rate{B}{\beta} \cost{B}{A} = 2$ (resp., $\rate{C}{\beta} \cost{C}{A} = 1$).
		
		\item Two nodes hold object $\alpha$ and the third holds $\beta$:
		In the case where $A$ and $B$ hold $\alpha$, $A$ wants to deviate since the cost while holding $\alpha$ and accessing $\beta$ from $A$'s nearest node $S$ is $\rate{A}{\beta} \cost{A}{S} = 1$; while the cost for holding $\beta$ and accessing $\alpha$ from $A$'s nearest node $B$ is $\rate{A}{\alpha} \cost{A}{B} = 0.7$.
		In the case where $A$ and $C$ hold $\alpha$, then $C$ wants to deviate since the cost while holding $\alpha$ and accessing $\beta$ from $C$'s nearest node $B$ is $\rate{C}{\beta} \cost{C}{B} = 2$; while the cost for holding $\beta$ and accessing $\alpha$ from $C$'s nearest node $A$ is $\rate{C}{\alpha} \cost{C}{A} = 1$.
		In the case where $B$ and $C$ hold $\alpha$, $B$ wants to deviate since the cost while holding $\alpha$ and accessing $\beta$ from $B$'s nearest node $A$ is $\rate{B}{\beta} \cost{B}{A} = 2$; while the cost for holding $\beta$ and accessing $\alpha$ from $B$'s nearest node $C$ is $\rate{B}{\alpha} \cost{B}{C} = 1$.
		
		\item One node holds $\alpha$: If $A$ (resp., $B$, or $C$) holds $\alpha$, $B$ (resp., $C$, $A$) wants to deviate since the cost while holding $\beta$ and accessing $\alpha$ from $B$'s (resp., $C$'s, or $A$'s) nearest node $A$ (resp., $B$, or $C$) is $\rate{B}{\alpha} \cost{B}{A} = 2$ (resp.,  $\rate{C}{\alpha} \cost{C}{B} = 2$, or  $\rate{A}{\alpha} \cost{A}{C} = 1.4$); while the cost for holding $\alpha$ and accessing $\beta$ from $B$'s (resp., $C$'s, or $A$'s) nearest node $C$ (resp., $A$, or $B$), is $\rate{B}{\beta} \cost{B}{C} = 1$ (resp., $\rate{C}{\beta} \cost{C}{A} = 1$, or $\rate{A}{\beta} \cost{A}{B} = 1$).
		
		\item Nodes $A$, $B$, and $C$ hold $\beta$: All of them want to deviate. Node $A$ wants to deviate since the cost while holding $\beta$ and accessing $\alpha$ from $A$'s nearest node $C_i$, for some $i \in [1,k]$, is $\rate{A}{\beta} \cost{A}{C_i} = 3$; while the cost for holding $\beta$ and accessing $\alpha$ from $A$'s nearest node $S$ is $\rate{A}{\alpha} \cost{A}{S} = 0.7$. Similar proof holds for nodes $B$ and $C$.
		
	\end{itemize}
	Obviously the system does not have a pure Nash equilibrium, which completes the proof.
\end{proof}

\begin{theorem}
\label{Theorem:metric}
The \icc\ instance $I_2$ has an equilibrium if and only if node $\node{S}$ holds object $\object{\alpha}$.
\end{theorem}

\begin{proof}
First, assume that $S$ is holding $\alpha$.  By Lemma~\ref{Lemma:S_node} nodes $C_1, \ldots, C_k$ hold object $\beta$, and by Lemma~\ref{Lemma:clause_node} at least one of nodes $\node{L_{i,j}}$, for $j \in \{1,2,3\}$ for each node $C_i, i \in [1,k]$, holds object $\alpha$, and the corresponding $\bar{L}_{i,j}$ is holding object $\beta$.  We claim that the placement where $A$ holds $\gamma$, node $B$ holds $\beta$, and $C$ holds $\gamma$ is a pure Nash equilibrium.  We prove this by showing that none of these nodes wants to deviate from their strategy.  Node $A$ doesn't want to deviate since the cost for holding object $\gamma$ and accessing object $\alpha$ from node $S$ is $\rate{A}{\alpha} \cost{A}{S} = 3$; while the cost for holding $\alpha$ and accessing $\gamma$ from node $C$ increases to $\rate{A}{\gamma} \cost{A}{C} = 4$.  Node $B$ doesn't want to deviate since the cost for holding object $\beta$ and accessing object $\gamma$ from node $C$ is $\rate{B}{\gamma} \cost{B}{C} = 3.000285$; while the cost for holding object $\beta$ and accessing $\gamma$ from the server increases to $\rate{B}{\gamma} \distsrv = 5$.  Node $\node{C}$ doesn't want to deviate since the cost for holding object $\gamma$ and accessing $\beta$ from node $B$ is $\rate{C}{\beta} \cost{C}{B} = 3.05$; while the cost for holding object $\beta$ and accessing $\gamma$ from node $A$ increases to $\rate{C}{\beta} \cost{C}{A} = 3.2$.
	
Now assume that node $\node{S}$ holds object $\object{\beta}$.  We are going to prove that for every possible placement over nodes $A$, $B$, and $C$, at least one node wants to deviate from its strategy.  Consider the following cases:
\begin{itemize}
	\item Node $A$ holds $\alpha$, node $B$ holds $\gamma$, and node $C$ holds $\beta$: Node $A$ wants to deviate since the cost while it is holding object $\alpha$ and accessing object $\gamma$ from node $B$ is $(\rate{A}{\gamma} \cost{A}{B} = 6.2)$; while the cost for holding object $\gamma$ and accessing $\alpha$ from the server decreases to $\rate{A}{\alpha} \distsrv = 5$.
			
	\item Node $A$ holds $\gamma$, node $B$ holds $\gamma$, and node $C$ holds $\beta$: Node $B$ wants to deviate since the cost while it is holding object $\gamma$ and accessing object $\beta$ from node $C$ is $(\rate{B}{\beta} \cost{B}{C} = 3.05)$; while the cost for holding object $\beta$ and accessing $\gamma$ from node $A$ decreases to $\rate{B}{\gamma} \cost{B}{A} = 3.04947$.
			
	\item Node $A$ holds $\gamma$, node $B$ holds $\beta$, and node $C$ holds $\beta$: Node $C$ wants to deviate since the cost while it is holding object $\beta$ and accessing object $\gamma$ from node $A$ is $(\rate{C}{\gamma} \cost{C}{A} = 3.2)$; while the cost for holding object $\gamma$ and accessing $\beta$ from node $B$ decreases to $\rate{C}{\beta} \cost{C}{B} = 3.05$.
			
	\item Node $A$ holds $\gamma$, node $B$ holds $\beta$, and node $C$ holds $\gamma$: Node $A$ wants to deviate since the cost while it is holding object $\gamma$ and accessing object $\alpha$ from the server is $(\rate{A}{\alpha} \distsrv = 5)$; while the cost for holding object $\alpha$ and accessing $\gamma$ from node $C$ decreases to $\rate{A}{\gamma} \cost{A}{C} = 4$.
			
	\item Node $A$ holds $\alpha$, node $B$ holds $\beta$, and node $C$ holds $\gamma$: Node $B$ wants to deviate since the cost while it is holding object $\beta$ and accessing object $\gamma$ from node $C$ is $(\rate{B}{\gamma} \cost{B}{C} = 3.000285)$; while the cost for holding object $\gamma$ and accessing $\beta$ from node $S$ decreases to $\rate{B}{\beta} \cost{B}{S} = 3$.
			
	\item Node $A$ holds $\alpha$, node $B$ holds $\gamma$, and node $C$ holds $\gamma$: Node $C$ wants to deviate since the cost while it is holding object $\gamma$ and accessing object $\beta$ from the server is $(\rate{C}{\beta} \distsrv = 5)$; while the cost for holding object $\beta$ and accessing $\gamma$ from $B$ decreases to $\rate{C}{\gamma} \cost{B}{C} = 4.88$.
	
	\item Node $A$ holds $\alpha$, node $B$ holds $\beta$, and node $C$ holds $\beta$: Node $C$ wants to deviate since the cost while it is holding object $\beta$ and accessing object $\gamma$ from the server is $(\rate{C}{\gamma} \distsrv = 4.9185)$; while the cost for holding object $\gamma$ and accessing $\beta$ from node $B$ decreases to $\rate{B}{\beta} \cost{C}{B} = 3.05$.
			
	\item Node $A$ holds $\gamma$, node $B$ holds $\gamma$, and node $C$ holds $\gamma$: Node $A$ wants to deviate since the cost while it is holding object $\gamma$ and accessing object $\alpha$ from the server is $(\rate{A}{\alpha} \distsrv = 5)$; while the cost for holding object $\alpha$ and accessing $\gamma$ from $C$ decreases to $\rate{A}{\gamma} \cost{A}{C} = 4$.
			
\end{itemize}
The remaining placements where $A$ holds $\alpha$, $B$ holds $\alpha$, and $C$ holds $\alpha$, obviously are not stable since none of the nodes are interested in these objects.  Since there does not exist a stable placement, an equilibrium does not exist.
\end{proof}

\BfPara{Binary object preferences over three objects}.  For the binary object preferences, we introduce two extra nodes $K$ and $L$. We set $\cost{C_i}{K}$, for $i \in [1,k]$, between clause nodes and $K$ to be $1.4$, $\cost{S}{L}$ to be $2.1$, and $\cost{A}{S}$, $\cost{A}{B}$, $\cost{B}{C}$, $\cost{C}{A}$ to be $1$.  The server node, which is at access cost $\distsrv = 10$ from all nodes in $\nodeSet$, stores a fixed copy of three objects $\alpha$, $\beta$, and $\gamma$.  Each node $i$ has a set $S_i$ of objects in which it is equally interested.  For nodes $X_i$, $\bar{X}_i$, for $i \in [1,n]$, we set $S_{X_i} = \{\alpha, \beta\}$ and $S_{\bar{X}_i} = \{\alpha, \beta\}$.  For nodes $C_i$, for $i \in [1,k]$, we set $S_{C_i} = \{\alpha, \gamma\}$.  For node $K$ we set $S_K = \{\gamma\}$; while for node $L$ we set $S_L = \{\beta\}$.  For node $S$ we set $S_S = \{\alpha, \beta\}$.  For nodes $A$, $B$, and $C$ we set $S_A$, $S_B$, and $S_C$ correspondingly to be the set $\{\alpha, \gamma\}$.  As we mentioned in the binary object preference definition for our utility function $\sumUtility{i}$, equally interested means weight $1$ for all objects in $S_i$, and $0$ for the remaining.  We refer to this instance as $I_3$.

Lemma~\ref{Lemma:var_node} holds as it is for the binary object preferences directed case.

\begin{lemma}
\label{Lemma:clause_node_01}
Clause node $\node{C_i}$ holds object $\object{\alpha}$ if and only if its variable nodes $\node{L_{i,j}}$, for $j \in \{1,2,3\}$ hold object $\object{\beta}$.
\end{lemma}

\begin{proof}
First, assume that $\node{L_{i,j}}$, for $j \in \{1,2,3\}$ hold $\object{\beta}$. By Lemma~\ref{Lemma:var_node} we know that nodes $\bar{L}_{i,j}$, for $j \in \{1,2,3\}$ hold $\alpha$, and they are $C_i$'s nearest nodes holding $\object{\alpha}$; while $C_i$'s nearest node holding $\gamma$ is node $K$.  Node's $C_i$ cost for holding $\object{\alpha}$ and accessing $\object{\gamma}$ from $K$ is $\cost{C_i}{K} = 1.4$; while the cost for holding $\object{\gamma}$ and accessing $\object{\alpha}$ from $\node{\bar{L}_{i,j}}$, for $j \in \{1,2,3\}$, is $\cost{C_i}{\bar{L}_{ij}} = 1.5$. Obviously, node $\node{C_i}$ prefers to replicate $\object{\alpha}$.
	
Now assume that at least one of the nodes $\bar{L}_{i,j}$, for $j \in \{1,2,3\}$ holds $\alpha$. These nodes are  $C_i$'s nearest nodes holding $\alpha$; while again $C_i$'s nearest node holding $\gamma$ is node $K$.  Node's $C_i$ cost for holding $\gamma$ and accessing $\object{\alpha}$ from $\node{L_{i,j}}$, for $j \in \{1,2,3\}$, is $\cost{C_i}{L_{i,j}} = 1$; while the cost for holding $\alpha$ and accessing $\gamma$ from node $K$ is $\cost{C_i}{K} = 1.4$.  Obviously, node $\node{C_i}$ prefers to replicate $\object{\gamma}$.
\end{proof}

\begin{lemma}
\label{Lemma:S_node_01}
Node $\node{S}$ holds object $\object{\alpha}$ if and only if all clause nodes $\node{C_1}, \ldots, \node{C_k}$ hold object $\object{\gamma}$.
\end{lemma}

\begin{proof}
	First, assume that $\node{C_1}, \ldots, \node{C_k}$ are holding $\gamma$. By Lemma~\ref{Lemma:clause_node_01}, $S$'s nearest node holding $\alpha$ is at least one of $\node{L_{i,j}}$ nodes, where $i \in [1,k], j \in \{1,2,3\}$; while $S$'s nearest nodes holding $\beta$ is node $L$. The cost for $S$ holding $\alpha$ and accessing $\beta$ from node $L$, is $\cost{S}{L} = 2.1$; while the cost for holding $\beta$ and accessing $\alpha$ from $\node{L_{i,j}}$, where $i \in [1,k], j \in \{1,2,3\}$, is $\cost{S}{L_{i,j}} = 3$. Obviously, node $\node{S}$ prefers to replicate $\alpha$.
	
	Now assume that at least one of $\node{C_1}, \ldots, \node{C_k}$ holds $\object{\alpha}$. These nodes are $S$'s nearest node holding $\alpha$; while again $S$'s nearest node holding $\beta$ is $L$. The cost for holding $\beta$ and accessing $\alpha$ from a node $\node{C_i}$, is $\cost{S}{C_i} = 2$; while the cost for holding $\alpha$ and accessing $\beta$ from a node $L$ is $\cost{S}{L} = 2.1$. Obviously, node $\node{S}$ prefers to replicate $\beta$.
\end{proof}

\begin{theorem}
\label{Theorem:nonmetric_01}
There exists an equilibrium for the \icc\ instance $I_3$ if and only
if node $S$ holds object $\alpha$.
\end{theorem}

\begin{proof}
	First, assume that $S$ is holding $\alpha$. By Lemma~\ref{Lemma:S_node_01} nodes $C_1, \ldots, C_k$ hold object $\gamma$, and by Lemma~\ref{Lemma:clause_node_01} at least one of nodes $\node{L_{i,j}}$, for $j \in \{1,2,3\}$ for each node $C_i, i \in [1,k]$, holds object $\alpha$, and the corresponding $\bar{L}_{i,j}$ is holding object $\beta$. We claim that the placement where $A$ holds $\gamma$, $B$ holds $\gamma$, and $C$ holds $\alpha$, is a pure Nash equilibrium. We prove this by showing that none of these nodes wants to deviate from their strategy.
	
	Node $A$ does not want to deviate since its cost for holding object $\gamma$ and accessing $\alpha$ from $A$'s nearest node $S$, is $\cost{A}{S} = 1$; while the cost for holding object $\alpha$ and accessing $\gamma$ from $A$'s nearest node $B$, is still $\cost{A}{B} = 1$.
	Node $B$ does not want to deviate since its cost for holding object $\gamma$ and accessing $\alpha$ from $B$'s nearest node $C$, is $\cost{B}{C} = 1$; while the cost for holding object $\alpha$ and accessing $\gamma$ from $B$'s nearest node $A$, is still $\cost{B}{A} = 1$.
	Node $C$ does not want to deviate since its cost for holding object $\alpha$ and accessing $\gamma$ from $C$'s nearest node $A$, is $\cost{C}{A} = 1$; while the cost for holding object $\gamma$ and accessing $\alpha$ from $C$'s nearest node $S$, is still $\cost{C}{S} = 1$.
	Also note that none of $S, C_1, \ldots, C_k, L_{ij}, \bar{L}_{ij}$ for $i \in [1,k], j \in \{1,2,3\}$ is getting affected of the objects been holded by the gadget nodes.\\
	
	Now assume that node $\node{S}$ holds object $\beta$. We are going to prove that for every possible placement over nodes $A$, $B$, and $C$, at least one node wants to deviate from its strategy. Consider the following cases:
	\begin{itemize}
		\item Nodes $A$, $B$, and $C$ hold object $\alpha$: Node $B$ (resp., $C$) wants to deviate, since the cost for holding object $\alpha$ and accessing $\gamma$ from $B$'s (resp., $C$'s) nearest node $C_i$, for some $i \in [1,k]$ or from node $K$, is $\cost{B}{C_i} = 5$ or $\cost{B}{K} = 6.4$ (resp., $\cost{C}{C_i} = 4$ or $\cost{C}{K} = 5.4$); while the cost for holding object $\gamma$ and accessing $\alpha$ from $B$'s nearest node $A$, is $\cost{B}{A} = 2$ (resp., $\cost{C}{A} = 1$).
		
		\item Two nodes hold object $\alpha$ and the third holds $\gamma$:
		In the case where $A$ and $B$ hold $\alpha$, $A$ wants to deviate since the cost while holding $\alpha$ and accessing $\gamma$ from $A$'s nearest node $C$ is $\cost{A}{C} = 2$; while the cost for holding $\gamma$ and accessing $\alpha$ from $A$'s nearest node $B$ is $\cost{A}{B} = 1$.
		The other cases are symmetric.
		
		\item One node holds $\alpha$: If $A$ holds $\alpha$, $B$ wants to deviate since the cost while holding $\gamma$ and accessing $\alpha$ from $B$'s nearest node $A$ is $\cost{B}{A} = 2$; while the cost for holding $\alpha$ and accessing $\gamma$ from $B$'s nearest node $C$, is $\cost{B}{C} = 1$. The other cases are symmetric.
		
		\item Nodes $A$, $B$, and $C$ hold $\gamma$: All of them want to deviate. Node $A$ wants to deviate since the cost while holding $\gamma$ and accessing $\alpha$ from $A$'s nearest node $C_i$, for some $i \in [1,k]$, is $\cost{A}{C_i} = 3$; while the cost for holding $\alpha$ and accessing $\gamma$ from $A$'s nearest node $B$ is $\cost{A}{B} = 1$. The other cases are symmetric.
		
	\end{itemize}
	Obviously the system does not have a pure Nash equilibrium, which completes the proof.
\end{proof}

We now show that $\phi$ is satisfiable if and only if the above \icc\
games (both undirected and directed cases) (resp., for the binary object
preferences, directed case) has a pure Nash equilibrium. Suppose
that $\phi$ is satisfiable and consider a satisfying assignment for
$\phi$.  If the assignment of a variable $x_i$ is True, then we
replicate object $\alpha$ in cache of variable node $X_i$; otherwise,
we replicate object $\beta$. By Lemma~\ref{Lemma:var_node} we know
that a variable node $X_i$ holds object $\object{\alpha}$ (resp.,
$\object{\beta}$) if and only if node $\bar{X}_i$ holds object $\beta$
(resp., $\alpha$). In this way we keep the consistency between truth
assignment of a variable and its negation.  By
Lemma~\ref{Lemma:clause_node} (resp.,
Lemma~\ref{Lemma:clause_node_01}) we know that a clause node $C_i$,
will replicate object $\beta$ (resp., $\gamma$) if and only if at
least one of its variable nodes, holds object $\alpha$. From above,
any clause node $C_i$ will hold object $\beta$ (resp., $\gamma$) only
if at least one of clause $c_i$ literals is True.  By
Lemma~\ref{Lemma:S_node} (resp., Lemma~\ref{Lemma:S_node_01}), we know
that node $\node{S}$, will replicate object $\object{\alpha}$ if and
only if all clause nodes $C_1, \ldots, C_k$ are holding object
$\object{\beta}$ (resp., $\gamma$). Thus, node $\node{S}$ replicates
object $\alpha$ only if all clauses $c_1, \ldots, c_k$ are True.  By
Theorems~\ref{Theorem:nonmetric} and \ref{Theorem:metric} (resp.,
\ref{Theorem:nonmetric_01}), we know that there exists a pure Nash
Equilibrium if and only if object $\object{\beta}$ is stored to node
$\node{S}$; thus, there exists a pure Nash Equilibrium if and only if
all clauses are True. This gives our proof.

\subsection{Binary preferences over two objects}
\label{sec:dir2obj}
Consider the problem \twobin: does a given \icc\ instance with two objects and binary preferences possess an equilibrium?  We prove that \twobin\ is polynomial-time equivalent to the
notorious \evencycle\ problem~\cite{Younger73}: does a given digraph
contain an even cycle?  Despite intensive efforts, the complexity of
the problem \evencycle\ was open until~\cite{McCuaigRST97, RST99}
provided a tour de force polynomial-time algorithm.  Our result thus also places \twobin\ in
\Poly.

\begin{theorem}
  \evencycle\ is polynomial-time equivalent to \twobin. 
\end{theorem}

We prove the polynomial-time equivalence of \twobin\ and \evencycle\
by a series of reductions.  We first show the equivalence between
\twobin\ and \twodirbin, which is the sub-class of \twobin\ instances
in which the node preferences are specified by an unweighted directed
graph (henceforth {\em digraph}); in a \twodirbin\ instance, we are
given a digraph, and the preference of a node for the other nodes
increases with decreasing distance in the graph.

\begin{lemma}
  \twobin\ is polynomial-time equivalent to \twodirbin.
\end{lemma}
\begin{proof}
  Given a \twobin\ instance $I$ with node set $\nodeSet$, two objects,
  node preference relations $\{\nodePref{i}: i \in \nodeSet\}$, and
  interest sets $\{S_i: i \in \nodeSet\}$, we construct a \twodirbin\
  instance $I'$ with the same node set, objects, and interest sets,
  but with the node preference relations specified by an unweighted
  digraph $G$.  Our construction will ensure that any equilibrium in
  $I$ is an equilibrium in $I'$ and vice-versa.  For distinct nodes
  $i$ and $j$, we have an edge from $i$ to $j$ if and only if $j$ is a
  most $i$-preferred node in $\nodeSet \setminus \{i\}$.  We now argue
  that $I$ has an equilibrium if and only if $I'$ has an equilibrium.
  A placement for $I$ is an equilibrium if and only if the following
  holds for each node $i$: (a) if $|S_i| = 1$, then $i$ holds the lone
  object in $S_i$; (b) if $|S_i| = 2$, then the object not held by $i$
  is at an $i$-most preferred node.  Similarly, any equilibrium
  placement for $I'$ satisfies the following condition for each $i$:
  (a) if $|S_i| = 1$, then $i$ holds the lone object in $S_i$; (b) if
  $|S_i| = 2$, then the object not held by $i$ is at a neighbor of
  $i$.  By our construction of the instances, equilibria of $I$ are
  equilibria of $I'$ and vice-versa.
\end{proof}

We next define \exacttwodirbin, which is the subclass of \twodirbin\
games where each node is interested in both objects; thus, an
\exacttwodirbin\ instance is completely specified by a digraph $G$.
We say that a node $i$ is {\em stable}\/ in a given placement $P$ if
$P_i$ is a best response to $P_{-i}$.  We say that an \exacttwodirbin\
instance $G$ is {\em stable}\/ (resp., {\em 1-critical}) if there
exists a placement in which all nodes (resp., all nodes except at most
one) are stable.  Since each node has unit cache capacity, each
placement is a 2-coloring of the nodes: think of a node as colored by
the object it holds in its cache.  Given a placement, an arc is said
to be bichromatic if its head and tail have different colors. Note
that for any \exacttwodirbin\ instance, a node is stable in a
placement iff it has a bichromatic outgoing arc.

\begin{lemma}
\label{lem:dir-bin-general}
\twodirbin\ and \exacttwodirbin\ are polynomial-time
equivalent on general digraphs.
\end{lemma}
\begin{proof}
  Since \exacttwodirbin\ games are a special subclass of
  \twodirbin\ games, we only need to show that \twodirbin\ games
  reduce to \exacttwodirbin\ games. Given an instance of a
  \twodirbin\ game, we need to handle the nodes that are interested in
  at most one object.  First, note that we can remove the outgoing
  arcs from all such nodes. Let $V_0$ consist of the nodes with no
  objects of interest.  For each node $u$ in $V_0$ we add a new node
  $u_0$ to $V_0$ along with arcs $(u, u_0)$ and $(u_0, u)$.  Let red
  and blue denote the two objects.  Let $V_r$ and $V_b$ denote the set
  of nodes interested in red and blue, respectively.  Without loss of
  generality, let $|V_r| \geq |V_b|$.  Add $|V_r| - |V_b|$ additional
  nodes to the set $V_b$ (so that $|V_r| = |V_b|$) and connect all the
  nodes in $V_r \bigcup V_b$ with a directed cycle that alternates
  strictly between $V_r$ nodes and $V_b$ nodes.  The rest of the
  network is kept the same and all the nodes are set to have interest
  in both objects.  Now, if the original instance is stable then we
  can stabilize the new instance by having each node in $V_r$ (resp.,
  $V_b$) cache the red (resp., blue) object, the nodes in $V_0$ cache
  any object (so long as an original node $u$ and its associated node
  $u_0$ store complementary objects) and the other nodes cache the
  same object as in the placement that made the original instance
  stable. And in the other direction, if the transformed instance is
  stable then in an equilibrium placement, the nodes in $V_r$ must
  each store an object of one color while each node in $V_b$ stores
  the object of the other color.  By renaming the colors, if
  necessary, we get a stable coloring (placement) for the original
  instance.
\end{proof}

For completeness, we next present some standard graph-theoretic
terminology that we will use in our proof.  A digraph is said to be
{\em weakly} connected if it is possible to get from a node to any
other by following arcs without paying heed to the direction of the
arcs.  A digraph is said to be {\em strongly} connected if it is
possible to get from a node to any other by a directed path.  We will
use the following well-known structure result about digraphs: a
general digraph that is weakly connected is a directed acyclic graph
on the unique set of maximal strongly connected (node-disjoint)
components.  We will also use the following strengthening of the
folklore ear-decomposition of strongly connected
digraphs~\cite{Schrijver}:
\begin{lemma}
\label{lem:ear-decomp}
An ear-decomposition can be obtained starting with any cycle of a strongly
connected digraph.
\end{lemma}
\begin{proof}
  The proof is by contradiction. Suppose not, then consider a subgraph
  with a maximal ear-decomposition obtainable from the cycle in
  question.  If it is not the entire digraph then consider any arc
  leaving the subgraph. Note that the digraph is strongly connected
  and hence such an arc must exist. Further, note that every arc in a
  digraph is contained in a cycle since there is a directed path from
  the head of the arc to the tail. Starting from the arc follow this
  cycle until it intersects the subgraph again, as it must because it
  ends at the tail which lies in the subgraph. This forms an ear that
  contradicts the maximality of the decomposition.
\end{proof}

\begin{lemma}
\label{lem:even-cycle}
\evencycle\ on strongly connected digraphs and \evencycle\ on general
digraphs are polynomial-time equivalent.
\end{lemma}
\begin{proof}
  Since strongly connected digraphs are a special subclass of general
  digraphs it suffices to show that \evencycle\ on general digraphs
  can be reduced to {\sc EVEN-CYCLE} on strongly connected digraphs.
  Remember that a general digraph has a unique set of maximal strongly
  connected components that are disjoint and computable in
  polynomial-time. Further any cycle, including even cycles, must lie
  entirely within a strongly connected component. Thus a digraph
  possesses an even cycle iff one of its strongly connected components
  does. Hence it follows that {\sc EVEN-CYCLE} on general digraphs reduces
  to \evencycle\ on strongly connected digraphs.
\end{proof}

\begin{lemma}
\label{lem:strong}
\evencycle\ and \exacttwodirbin\ games are polynomial-time equivalent
on strongly connected digraphs.
\end{lemma}
\begin{proof}
  To show the polynomial-time equivalence, we show that a strongly
  connected digraph is stable iff it has an even cycle. One direction
  is easy.  If the digraph is stable then consider the placement in
  which every node is stable.  So every node has a bichromatic
  outgoing arc; by starting at any node and following outgoing
  bichromatic edges we will eventually loop back on ourselves. The
  loop so obtained is the required even cycle; it is even because it
  is composed of bichromatic arcs.  In the other direction, if there
  is an even cycle then we take the ear-decomposition starting with
  that cycle (Lemma~ \ref{lem:ear-decomp}), stabilize that cycle (by
  making each arc bichromatic since it is of even cardinality) and
  then stabilize each node in each ear by working backwards along the
  ear.
\end{proof}

\begin{lemma}
\label{lem:1critical}
Any \exacttwodirbin\ game on a strongly connected digraph is 1-critical. 
\end{lemma}
\begin{proof}
  Consider an ear-decomposition of the strongly connected digraph
  starting with a cycle.  Observe that all but at most one node of the
  cycle can be stabilized by arbitrarily assigning one color to a
  node, and then assigning alternate colors to the nodes as we
  progress along the cycle.  Every node in the cycle, other than
  possibly the initial node, is stable.  The rest of the digraph can
  be stabilized ear by ear, stabilizing each ear by working backwards
  from the point of attachment.  Hence, all but one node of the digraph
  can be stabilized.
\end{proof}

\begin{lemma}
\label{lem:dir-bin-exact}
\exacttwodirbin\ on general digraphs is polynomial-time equivalent to
\exacttwodirbin\ on strongly connected digraphs.
\end{lemma}
\begin{proof}
  Since strongly connected digraphs are a subclass of general digraphs
  we need only show that the problem \exacttwodirbin\ on general
  digraphs reduces to \exacttwodirbin\ on strongly connected digraphs.
  A general digraph is stable iff all of its weakly connected
  components are. A weakly connected component is a directed acyclic
  graph (dag) on the strongly connected components.  It is clear that
  a weakly connected component cannot be stabilized if any one of the
  strongly connected components that is a minimal element of the
  directed acyclic graph cannot be stabilized.  Interestingly, the
  converse is also true. If all of the strongly connected components
  that are minimal elements of the dag can be stabilized then the
  entire weakly connected component can be stabilized because each of
  the other strongly connected components has at least one outgoing
  arc which is used to stabilize its tail while the rest of the
  strongly connected component can be stabilized because strongly
  connected components are 1-critical by Lemma~\ref{lem:1critical}.
  We can determine such a stable placement by processing the strongly
  connected components in topologically sorted order (according to the
  dag) starting from the minimal elements.  Thus a digraph is stable
  iff every strongly connected component that is a minimal element is
  stable. Hence, \exacttwodirbin\ on general digraphs is reducible in
  polynomial-time to strongly connected digraphs.
\end{proof}

\section{Fractional replication games}
\label{sec:frac}
We introduce a new class of capacitated replication games where nodes
can store fractions of objects, as opposed to whole objects, and
satisfy an object access request by retrieving enough fractions that
make up the whole object.  Rather than associate different identities
with different fractions of a given object, we view each portion of an
object as being fungible, thus allowing any set of fractions of an
object, adding up to at least one, to constitute the whole object.
Such fractional replication scenarios naturally arise when objects are
encoded and distributed within a network to permit both efficient and
reliable access.

Several implementations of fractional replication, in fact, already
exist.  For instance, fountain codes~\cite{Byers98,Shok06} and the
information dispersal algorithm~\cite{Rabin89} present two ways of
encoding an object as a number of smaller pieces -- of size, say $1/m$
fraction of the full object size, where $m$ is an integer -- such that
the full object may be reconstructed from any $m$ of the pieces.  A
natural formalization is to view each object as a polynomial of high
degree, and consider each piece of the object as the evaluation of the
polynomial on a random point in a suitable large field.  Then,
accessing an object is equivalent (with very high probability) to
accessing a sufficient number of pieces of the object.

We now present fractional capacitated selfish replication (\fcc)
games, which are an adaptation of the game-theoretic framework
developed in Section~\ref{sec:model} to fractional replication.  We
have a set $\nodeSet$ of nodes sharing a set $\objSet$ of objects.  In
an \fcc\ game, the strategies are {\em fractional placements}; a
fractional placement $\widetilde{P}$ is a $|V|$-tuple
$\{\widetilde{P}_i: i \in V\}$ where $\widetilde{P}_i: \objSet
\rightarrow \Re$ under the constraint that sum of
$\widetilde{P}_i(\alpha)$, over all $\alpha$ in $\objSet$, is at most
the cache size of $i$.

We begin by presenting \fcc\ games in the special case
of sum utilities, where the generalization from the integral to the
fractional setting is most natural.
For sum utilities, recall that we are given a cost function
$\costNoArgs$ and node-object weights $r_i(\alpha)$, $i \in \nodeSet$,
$\alpha \in \objSet$.  Given a fractional global placement
$\widetilde{P}$, we define the cost incurred by $i$ for accessing
object $\alpha$ as the minimum value of $x_j \cost{i}{j}$ under the
constraints that $\sum_{j} x_j = 1$ and $x_j \le
\widetilde{P}_j(\alpha)$ for all $j$.  Then, the total cost incurred
by $i$ is the sum, over all objects $\alpha$, of $r_i(\alpha)$ times
the cost incurred by $i$ for accessing $\alpha$.  For a given
fractional global placement $\widetilde{P}$, the utility of $i$ is the
negative of the total cost incurred by $i$ under $\widetilde{P}$.

We now consider \fcc\ games under the more general setting of utility
preference relations.  As before, each node $i$ has a node preference
relation $\nodePref{i}$ and a preference relation $\placePref{i}$
among global (integral) placements.  Recall that the node and
placement preference relations of each node $i$ induce a preorder
$\pairPref{i}$ among the elements of $\objSet \times (\nodeSet
\setminus \{i\})$ (see Section~\ref{sec:model}).  For \fcc\ games, we
require the existence of a {\em total}\/ preorder $\pairPref{i}$, for
all $i$.  We now specify the best response function for each player
for a given fractional global placement $\fP$.  For each node $i$ and
object $\alpha$, we determine the assignment $\opt_{i,\fP,\alpha}:
\nodeSet \setminus \{i\} \rightarrow \Re$ that is lexicographically
minimal under the node preference relation $\nodePref{i}$ subject to
the condition that $\opt_{i,\fP,\alpha} \le \fP_k(\alpha)$ for each
$k$ and $\sum_k \opt_{i,\fP,\alpha}(k) = 1$.  We next compute
$\best_{i,\fP} : \objSet \times (\nodeSet \setminus \{i\}) \rightarrow
\Re$ to be the lexicographically maximal assignment under
$\pairPref{i}$ subject to the condition that $\best_{i,\fP} (\alpha,k)
\le \opt_{i,\fP, \alpha}(k)$ for all $k$ and $\sum_{\alpha, k}
\best_{i,\fP} (\alpha,k)$ is at most the size of $i$'s cache.  The
best response of a player $i$ is then to store $\sum_k \best_{i,\fP}
(\alpha, k)$ of $\alpha$ in their cache.  This completes the
definition of \fcc\ games.

Using standard fixed-point machinery, we show that every \fcc\ game
has an equilibrium.  We also show that finding equilibria in
\fcc\ games is \ppad-complete.

\begin{theorem}
\label{thm:fcc}
Every \fcc\ instance has a pure Nash equilibrium.  Finding an
equilibrium in an \fcc\ game is \ppad-complete.
\end{theorem}

We prove Theorem~\ref{thm:fcc} by establishing separately the
existence of equilibria, membership in \ppad, and the \ppad-hardness
of finding equilibria.

\subsection{Existence of equilibria}
\begin{theorem}
\label{thm:fcc.exist}
Every \fcc\ instance has a pure Nash equilibrium.
\end{theorem}
\begin{proof}
By~\cite{OsborneRubinstein94} (Proposition 20.3, based on Kakutani’s fixed-point theorem), a game has a pure Nash equilibrium if the strategy
space of each player is a compact, non-empty, convex space, and the
payoff function of each player is continuous on the strategy space of
all players and quasi-concave in the strategy space of the player.  In
an \fcc\ instance, the strategy space of each player $i$ is simply the
set of all its fractional placements: that is, the set of functions $f
: \objSet \rightarrow [0,1]$ subject to condition that $\sum_{\alpha
  \in \objSet} f(\alpha) \le c_i$, where $c_i$ is the cache size of
the node (player).  The strategy set thus is clearly convex,
non-empty, and compact.  Furthermore, as defined above, the payoff for
any player $i$ under fractional placement $\widetilde{P}$ is simply
the solution to the following linear program:
\begin{eqnarray*}
\max - \sum_{\alpha \in \objSet} r_i(\alpha) (\sum_{j \in \nodeSet} x_{ij}(\alpha) \cost{i}{j}) & \\
\sum_{j \in \nodeSet} x_{ij}(\alpha)  = 1 & \mbox{ for all } i \in \nodeSet, \alpha \in \objSet\\
x_{ij}(\alpha) \le \widetilde{P}_j(\alpha) & \mbox{ for all } i,j \in \nodeSet, \alpha \in \objSet\\
x_{ij}(\alpha) \ge 0 & \mbox{ for all } i,j \in \nodeSet, \alpha \in \objSet
\end{eqnarray*}
It is easy to see that the payoff function is both continuous in the
placements of all players, and quasi-concave in the strategy space of
player $i$, thus completing the proof of the theorem.
\end{proof}

\subsection{Membership in \ppad}
\begin{theorem}
\label{thm:fcc.ppad}
Finding an equilibrium in an \fcc\ game is in \ppad.
\end{theorem}
\begin{proof}
Our proof is by a reduction from \fspp\ (Fractional Stable Paths
Problem), which is defined as follows~\cite{KintaliPRST09_focs}.  Let
$G$ be a graph with a distinguished destination node $d$.  Each node
$v \neq d$ has a list $\pi(v)$ of simple paths from $v$ to $d$ and a
preference relation $\geqprefer_v$ among the paths in $\pi(v)$.  For a path $S$, we also define $\pi(v, S)$ to be the set of paths in $\pi(v)$ that have $S$ as a suffix. A
\emph{proper suffix} $S$ of $P$ is a suffix of $P$ such that $S \neq
P$ and $S \neq \emptyset$.  A {\em feasible fractional paths
  solution}\/ is a set $w = \{w_v: v \neq d\}$ of assignments $w_v:
\pi(v) \rightarrow [0,1]$ satisfying:
(1) {\bf Unity condition}: for each node $v$, $\sum_{P \in \pi(v)} w_v(P) \le 1$, and (2) {\bf Tree condition}: for each node $v$, and each path $S$ with start node $u$, $\sum_{P \in
\pi(v, S)} w_v(P) \le w_u(S)$.

In other words, a feasible solution is one in which each node chooses
at most 1 unit of flow to $d$ such that no suffix is filled by more
than the amount of flow placed on that suffix by its starting node. A
feasible solution $w$ is {\em stable}\/ if for any node $v$ and path
$Q$ starting at $v$, one of the following holds:
{\bf (S1)} $\sum_{P \in \pi(v)} w_v(P) = 1$,
and for each $P$ in $\pi(v)$ with $w_v(P) > 0$, $P \geqprefer_v
Q$; or {\bf (S2)} There exists a proper suffix $S$ of $Q$ such
that $\sum_{P \in \pi(v,S)} w_v(P) = w_u(S)$, where $u$ is the start
node of $S$, and for each $P \in \pi(v,S)$ with $w_v(P) > 0$,
$P\geqprefer_v Q$.

Given an \fcc\ $G$ with node set $\nodeSet$, object set $\objSet$,
node preference relations $\nodePref{i}$ for $i \in \nodeSet$, and
utility preference relations $\placePref{i}$ for $i \in \nodeSet$, we
construct an instance $\FSPPInstance$ of \fspp\ as follows.  For nodes
$i,j \in \nodeSet$ and object $\alpha \in \objSet$, we introduce the
following \fspp\ vertices.
\begin{itemize}
\item
$\hold(i,\alpha)$ representing the amount of $\alpha$ that node $i$
  will store in its cache.
\item
$\serve(i,j,\alpha)$ representing the amount of $\alpha$ that node
  $j$ will serve for $i$ given a placement for $V \setminus \{i\}$.

\item
$\servePrime(i,j,\alpha)$, an auxiliary vertex needed for
  $\serve(i,j,\alpha)$.

\item
$\serve(i,\alpha)$, representing the amount of $\alpha$ that other
  nodes will serve for $i$ given a placement for $V \setminus \{i\}$.

\item
$\hold(i)$, representing the best response of $i$ give the placement of other nodes.

\item
$\holdPrime(i,\alpha)$, an auxiliary vertex needed for
  $\hold(i,\alpha)$.
\end{itemize}

We now present the path sets and preferences for each vertex of the
\fspp\ instance.

\begin{itemize}
\item
$\serve(i,\alpha)$: the path set includes all paths of the form
  $\langle \serve(i,\alpha), \hold(j,\alpha), d \rangle$, and
  $\serve(i,\alpha)$ prefers $\langle \serve(i,\alpha)$,$
  \hold(j,\alpha), d \rangle$ over $\langle \serve(i,\alpha),
  \hold(k,\alpha), d \rangle$ if $j \nodePref{i} k$.

\item
$\servePrime(i,j,\alpha)$: the path set includes all paths of the form
  $\langle$ $\servePrime(i,j,\alpha)$, $\serve($ $i,\alpha)$, $\hold(j,\alpha),
  d \rangle$ and the direct path $\langle \servePrime(i,j,\alpha),
  d\rangle$.  For the preference order, $\servePrime(i,j,\alpha)$
  prefers all paths $\langle \servePrime(i,j,\alpha),\serve(i,\alpha),
  \hold(j,\alpha), d \rangle$ equally, and all of them over the direct
  path.

\item
$\serve(i,j,\alpha)$: the path set includes the path $\langle
  \servePrime(i,j,\alpha),d\rangle$ and the direct path $\langle
  \serve(i,j,\alpha), d\rangle$ with a higher preference for the
  former path.

\item $\hold(i)$: the path set includes paths of the form $\langle
  \hold(i), \serve(i,j,\alpha), d\rangle$, and $\hold(i)$ prefers the
  path $\langle \hold(i), \serve(i,j,\alpha), d\rangle$ over $\langle
  \hold(i), \serve(i,k,\beta), d\rangle$ if $(j,\alpha) \pairPref{i}
  (k,\beta)$.

\item
$\holdPrime(i,\alpha)$: the path set includes paths of the form
  $\langle$ $\holdPrime(i,\alpha)$, $\hold(i)$, $\serve($ $i,j,\alpha)$, $d\rangle$ all of which are preferred equally, and the direct path
  $\langle \holdPrime(i,\alpha), d\rangle$ which is preferred the
  least.

\item
$\hold(i,\alpha)$: the path set includes two paths $\langle
  \holdPrime(i,\alpha), d\rangle$ and the direct path with a higher
  preference for the former path.
\end{itemize}

We now show that the \fcc\ instance has an equilibrium if and only if
the \fspp\ instance has an equilibrium.  Our proof is by giving a
mapping $f$ from global fractional placements in the \fcc\ instance to
feasible solutions in the \fspp\ instance such that (a) if $\fP$ is an
equilibrium for the \fcc\ instance, then $f(\fP)$ is an equilibrium
for the \fspp\ instance, and (b) if $w$ is an equilibrium for the
\fspp\ instance, then $f^{-1}(w)$ is an equilibrium for the
\fcc\ instance.

Let $\fP$ denote any fractional placement of the \fcc\ instance.  We
now define the solution $f(\fP)$ of the \fspp\ instance.  In $f(\fP)$
vertex $\hold(i,\alpha)$ plays $\fP_i(\alpha)$ on the direct path and
$1 - \fP_i(\alpha)$ on the other path in its path set, for every $i$
in $\nodeSet$ and $\alpha$ in $\objSet$.  The remaining vertices play
their best responses, considered in the following order.  First,
consider vertices of the form $\serve(i,\alpha)$.  In the best
response, the amount played by $\serve(i,\alpha)$ on the path $\langle
\serve(i,\alpha)$, $\hold(j,\alpha)$, $d\rangle$, equals
$\opt_{i,\fP,\alpha}(j)$; recall that $\opt_{i,\fP,\alpha}(j)$ is
the assignment that is lexicographically minimal under the node
preference relation $\nodePref{i}$ subject to the condition that
$\opt_{i,\fP,\alpha} \le \fP_k(\alpha)$ for each $k$ and $\sum_k
\opt_{i,\fP,\alpha}(k) = 1$.  We next consider the vertices of the
form $\servePrime(i,j,\alpha)$.  In its best response, vertex
$\servePrime(i,j,\alpha)$ plays $\opt_{i,\fP,\alpha}(j)$ on the path
$\langle \servePrime(i,j,\alpha)$, $\serve(i,\alpha)$, $\hold(j,\alpha),
d\rangle$.  Next, in its best response, vertex $\serve(i,j,\alpha)$
plays $\opt_{i,\fP,\alpha}(j)$ on its direct path and $1 -
\opt_{i,\fP,\alpha}(j)$ on its remaining path.  We now consider the
best response of vertex $\hold(i)$; it distributes its unit among
paths of the form $\langle \hold(i)$, $\serve(i,j,\alpha)$, $d\rangle$
(for all $j$ in $\nodeSet \setminus \{i\}$ and $\alpha$ in $\objSet$)
lexicographically maximally under the total preorder $\pairPref{i}$ over
node-object pairs.  That is, $\hold(i)$ plays
$\best_{i,\fP}(\alpha,j)$ on the path $\langle \hold(i),
\serve(i,j,\alpha), d\rangle$.  We next consider the best response of
the vertex $\holdPrime(i,\alpha)$; it plays $1 - \sum_j
\best_{i,\fP}(\alpha,j)$ on its direct path and
$\best_{i,\fP}(\alpha,j)$ on the path $\langle \holdPrime(i,\alpha),
\hold(i), \serve(i,j,\alpha), d\rangle$.  This completes the
definition of the solution $f(\fP)$.

We now argue that if $\fP$ is an equilibrium so is $f(\fP)$.  By
construction, every vertex other than of the form $\hold(i,\alpha)$
play their best responses in $f(\fP)$.  We next show that $i$ plays a
best response in $\fP$ if and only if the vertices $\hold(i,\alpha)$
play their best response in $f(\fP)$.  The best response of
$\hold(i,\alpha)$ is to play $1 - \sum_j \best_{i,\fP}(\alpha,j)$ on
the path $\langle \hold(i,\alpha), \holdPrime(i,\alpha), d\rangle$ and
the $\sum_j \best_{i,\fP}(\alpha,j)$ on its direct path.  The best
response of $i$ in $\fP$ is to set $\fP_i(\alpha)$ to $\sum_j
\best_{i,\fP}(\alpha,j)$.  Thus if $\fP$ is an equilibrium, then so is
$f(\fP)$.  Furthermore, if $w$ is an equilibrium, by definition of
$f$, $\fP = f^{-1}(w)$ is well-defined.  Since the best responses of
$i$ and the vertices $\hold(i,\alpha)$ are consistent, $\fP$ is also
an equilibrium.  This completes the reduction from \fcc\ to \fspp,
placing \fcc\ in \ppad.
\end{proof}

\subsection{PPAD-Hardness}
This section is devoted to the proof of the following theorem.
\begin{theorem}
\label{thm:ppad-hard}
	The problem of finding an equilibrium in \fcc\ games is PPAD-hard even when the underlying cost function $\costNoArgs$ is a metric.
\end{theorem}

Our reduction is from preference games~\cite{KintaliPRST09_focs}.
Given a preference game $G$ with $n$ players $1, 2, \ldots, n$ and
their preferences given by $\prefGamePref{i}$, we construct an
\fcc\ game $\widehat{G}$ as follows.  The game $\widehat{G}$ has a set
$V$ of $n^2 + 3n$ players numbered $1$ through $n^2 + 3n$, and a set
$\objSet$ of $2n$ objects $\alpha_1, \ldots, \alpha_{2n}$.  We set the
utility function for each node to be the sum utility function, thus
ensuring that the desired monotonicity and consistency conditions are
satisfied.

We next present the metric cost function $\costNoArgs$ over the nodes.
We group the players into four sets $V_1 = \{i: 1 \le i \le n\}$, $V_2
= \{i\cdot n + j: 1 \le i \le n, 1 \le j \le n\}$, $V_3 = \{n^2 + n +
i: 1 \le i \le n\}$, and $V_4 = \{n^2 + 2n + i: 1 \le i \le n\}$.  For
each node $i$ in $V_1$ and $j$ in $V_3$, we set $d_{ii} = 2$ and
$d_{ij} = 4$.  We set $d_{n^2 + n + i,n^2 + 2n + i} = 3$.  For each
node $i$ in $V_1$ and $k = i \cdot n + j$, we set $d_{ik}$ as follows:
if $j \prefGamePrefStrict{i} i$ then $d_{ij}$ equals $6 - \ell/n$ when
$j$ is the $\ell$th most preferred player for $i$; if $i
\prefGamePref{i} j$, then $d_{ij}$ equals $1$.  All the other
distances are obtained by using metric properties.

We finally specify the object weights.  For $k \in V_1$, we set
$r_k(\alpha_i) = 1$ for all $i \neq k$ such that $i \prefGamePref{k}
k$; we set $r_k(\alpha_k) = 2.5$ such that $4 < 2r_k(\alpha_k) \le 5$.
For node $k = i\cdot n + j$ in $V_2$, we set $r_k(\alpha_j) = 1$.  For
node $k = n^2 + n + i$ in $V_3$, we set $r_k(\alpha_i) =
r_k(\alpha_{i+n}) = 1$.  Finally, for node $k = n^2 + 2n +i$ in $V_4$,
we set $r_k(\alpha_{i+n}) = 1$.

Given a placement $P$ for $\widehat{G}$, we define a solution
$\omega(P) = \{w_{ij}\}$ for the preference game $G$: $w_{ij} =
P_i(\alpha_j)$.  The following lemma immediately follows from the
definition of $\widehat{G}$.
\begin{lemma}
\label{lem:frac.best}
The following statements hold for any placement $P$ for $\widehat{G}$.
\begin{itemize}
\item For $k = i\cdot n + j$, $1 \le j \le n$, $P_k$ is a best
  response to $P_{-k}$ if and only if $P_k(\alpha_j) = 1$.

\item For $k = n^2 + n + i$, $1 \le i \le n$, $P_k$ is a best response
  to $P_{-k}$ if and only if $P_k(\alpha_{n+i}) = 1$.

\item For $k = n^2 + n+ i$, $P_k$ is a best response to $P_{-k}$ if
  and only if $P_k(\alpha_i) = 1 - P_i(\alpha_i)$ and
  $P_k(\alpha_{n+i} = P_i(\alpha_i)$.
\end{itemize}
\end{lemma}

\begin{lemma}
\label{lem:frac.V_1.best}
Let $P$ be a placement for $\widehat{G}$ in which every node not in
$V_1$ plays their best response.  Then, the best response of a node
$i$ in $V_1$ is the lexicographically maximum $(P_i(\alpha_{j_1}),
P_i(\alpha_{j_2}), \ldots, P_i(\alpha_{j_n}))$, where $j_1
\prefGamePref{i} j_2 \prefGamePref{i} \cdots \prefGamePref{i} j_n$,
subject to the constraint that $P_i(\alpha_j) \le P_j(\alpha_j)$ for
$j \neq i$. \qed
\end{lemma}
\begin{proof}
Fix a node $i$ in $V_1$.  By Lemma~\ref{lem:frac.best}, node $i\cdot
n + j$ holds object $j$, for $1 \le j \le n$; each of these nodes is
at distance at least 5 and at most 6 away from $i$.  By
Lemma~\ref{lem:frac.best}, for every node $k = n^2 + n + j$, $1 \le j
\le n$, $P_k(\alpha_j) = 1 - P_j(\alpha_j)$ and $P_k(\alpha_{n+j}) =
P_j(\alpha_j)$.

We now consider the best response of node $i$.  We first note that for
any $j \in \{1,\ldots,n\} \setminus \{i\}$ such that $i
\prefGamePref{i} j$, $P_i(\alpha_j) = 0$ since the nearest full copy
of $\alpha_j$ is nearer than the nearest node holding any fraction of
object $\alpha_i$.  Let $S$ denote the set of $j$ such that $j
\prefGamePref{i} i$.  For any $j$ in $S \setminus \{i\}$,
$P_i(\alpha_j) \le P_j(\alpha_j)$ since node $n^2 + n + j$ at distance
$5$ holds $1 - P_j(\alpha_j)$ fraction of $\alpha_j$, the nearest node
holding any fraction of $\alpha_i$ is at distance 4, and
$4r_i(\alpha_i) > 5r_i(\alpha_j)$.  Furthermore, for any $j, k$ in $S$
if $j \prefGamePrefStrict{i} k$, then the farthest $P_j(\alpha_j)$
fraction of $\alpha_j$ is farther than the farthest $P_k(\alpha_k)$
fraction of $\alpha_k$, implying that in the best response, if
$P_i(\alpha_j) < P_j(\alpha_j)$ then $P_i(\alpha_k) = 0$.  Thus, the
best response of $i$ is the unique lexicographically maximum solution
$(P_i(\alpha_{j_1}), P_i(\alpha_{j_2}), \ldots, P_i(\alpha_{j_n}))$,
where $j_1 \prefGamePref{i} j_2 \prefGamePref{i} \cdots
\prefGamePref{i} j_n$, subject to the constraint that $P_i(\alpha_j)
\le P_j(\alpha_j)$ for $j \neq i$.
\end{proof}

\begin{lemma}
\label{lem:frac.reduction}
A placement $P$ is an equilibrium for $\widehat{G}$ if and only if
$\omega(P)$ is a equilibrium for $G$ and every node not in $V_1$ plays
their best response in $P$.
\end{lemma}
\begin{proof}
Consider an equilibrium placement $P$ for $\widehat{G}$ Clearly, every node
plays their best response.  We now prove that $\omega(P)$ is an equilibrium
for $G$.  Fix a node $i$ in $V_1$.  By Lemma~\ref{lem:frac.V_1.best},
the best response of $i$ is the unique lexicographically maximum
solution $(P_i(\alpha_{j_1}), P_i(\alpha_{j_2}), \ldots,
P_i(\alpha_{j_n}))$, where $j_1 \prefGamePref{i} j_2 \prefGamePref{i}
\cdots \prefGamePref{i} j_n$, subject to the constraint that
$P_i(\alpha_j) \le P_j(\alpha_j)$ for $j \neq i$.  Since this applies
to every node $i$, it is immediate from the definitions of $\omega(P)$
and preference games that if $P$ is an equilibrium for $\widehat{G}$ then
$\omega(P)$ is an equilibrium for $G$.

We now consider the reverse direction.  Suppose we have a placement
$P$ in which every player not in $V_1$ plays their best response and
$\omega(P)$ is an equilibrium for the preference game $G$.  By
Lemma~\ref{lem:frac.V_1.best} and the definition of $\omega(P)$, the
best response of $i$ in $G$ matches that in the \fcc\ game; hence
every player in $V_1$ also plays their best response in $P$, implying
that $P$ is an equilibrium for $\widehat{G}$.
\end{proof}

The construction of $\widehat{G}$ from $G$ is clearly polynomial time.
Furthermore, given any equilibrium for $\widehat{G}$, an equilibrium
for $G$ can be constructed in linear time.  We thus have a reduction
from a \ppad-complete problem to \fcc\ implying that the latter is
\ppad-hard, thus completing the proof of Theorem~\ref{thm:ppad-hard}.

\section{Concluding remarks} \label{sec:concs}
In this paper, we first define the integral and fractional selfish replication games (\cc\ and \fcc) in networks.  In our setup each node has a bounded cache capacity for uniform size objects.  We prove that every hierarchical network has a pure Nash equilibrium, introducing the notion of fictional players.  We almost completely characterize the complexity of \cc\ games, i.e. which classes have an equilibrium, the complexity of determine whether it exists, and if so, how efficiently it can be found.  For the open complexity question about undirected networks with binary preferences (proved to be potential games), we conjecture that finding equilibria is \pls-hard.  For the cases of games where equilibria exist, we study the convergence of the best response process.  The main focus of this work is in equilibria, leaving the problem of estimating the price of anarchy for all the configurations as future work, extending the work of~\cite{EB2017}.

We also show that \fcc\ games always have equilibria, though they may be hard to find.  It is not hard to argue that an equilibrium in the corresponding integral variant is an equilibrium in the fractional instance.  So whenever an ``integral'' equilibrium can be determined efficiently, so can a ``fractional'' equilibrium.  An interesting direction of research is to identify other special cases of fractional games where equilibria may be efficiently determined.  We also note that our proof of existence of equilibria in \fcc\ games, currently presented for the case of unit-size objects, extends to arbitrary object sizes.

Finally, even though our proofs work for a model that the sets of nodes, objects, and preference relations are all static, we believe that our results will be meaningful for dynamically changing environments.  Developing better models for addressing infrequently changes is a very important practical research direction.

\bibliographystyle{spbasic}      
\bibliography{thebibliography}   

\end{document}